\documentclass[journal]{IEEEtran}
\IEEEoverridecommandlockouts

\usepackage{caption}
\usepackage{graphicx}  

\usepackage{overpic}
\usepackage{amsmath,amssymb,amsfonts}
\usepackage{graphicx}
\usepackage{textcomp}
\usepackage{xcolor}
\usepackage{float}
\usepackage{amsthm}
\usepackage{graphicx}
\usepackage{epstopdf}
\usepackage{amsmath,bm}
\usepackage{amsfonts}
\usepackage{amssymb}
\usepackage{color}
\usepackage{multirow}
\usepackage{multicol}
\usepackage{soul,xcolor}
\usepackage{algorithm}
\usepackage{algpseudocode}%
\usepackage{float}
\usepackage{graphicx}    
\usepackage{subcaption}  

\theoremstyle{plain}
\newtheorem{theorem}{Theorem}

\newtheorem{lemma}{Lemma}

\newtheorem{corollary}{Corollary}
\newtheorem{remark}{Remark}

\newcommand{\vect}[1]{\mathbf{#1}}

\def\Htran{\mbox{\tiny $\mathrm{H}$}}
\def\Ttran{\mbox{\tiny $\mathrm{T}$}}
\def\CN{\mathcal{N}_{\mathbb{C}}} 
\def\imagunit{\mathsf{j}} 

\def\sinc{\mathrm{sinc}}

\usepackage{xcolor}
\usepackage{etoolbox}

\makeatletter
\pretocmd\@bibitem{\color{black}}{}{\fail}
\pretocmd\@lbibitem{\color{black}}{}{\fail}

\apptocmd\@bibitem{\@ifundefined{keycolor#1}{}{\csname keycolor#1\endcsname}}{}{\fail}
\apptocmd\@lbibitem{\@ifundefined{keycolor#2}{}{\csname keycolor#2\endcsname}}{}{\fail}

\newcommand\citecolor[2][blue]{\@namedef{keycolor#2}{\color{#1}}}
\makeatother

\begin{document}

\title{Near-Field Beamfocusing, Localization, and Channel Estimation With Modular Linear Arrays}

\author{Alva~Kosasih, \"Ozlem Tu\u{g}fe Demir, and~Emil~Bj{\"o}rnson,~\IEEEmembership{Fellow,~IEEE}  
\thanks{Parts of this paper were presented at the 2024 IEEE Asilomar Conference on Signals, Systems, and Computers \cite{Kosasih2024_Asilomar}. \\
\quad A. Kosasih was with KTH Royal Institute and now is with Nokia Standards, Espoo, Finland. Email: alva.kosasih@nokia.com. 
\"O. T. Demir was with TOBB University of Economics and Technology and now is with Bilkent University, Ankara, Turkiye. Email: ozlemtugfedemir@bilkent.edu.tr. 
E. Bj{\"o}rnson is with the Division of Communication Systems,  KTH Royal Institute of Technology, Stockholm, Sweden. E-mail: emilbjo@kth.se. \\
\quad This work was supported by the FFL18-0277 grant from the Swedish Foundation for Strategic Research and the Grant 2022-04222 from the Swedish Research Council.}}

\maketitle

\begin{abstract}
This paper investigates how near-field beamfocusing can be achieved using a modular linear array (MLA), composed of multiple widely spaced uniform linear arrays (ULAs). The MLA architecture extends the aperture length of a standard ULA without adding additional antennas, thereby enabling near-field beamfocusing without increasing processing complexity. Unlike conventional far-field beamforming, near-field beamfocusing enables simultaneous data transmission to multiple users at different distances in the same angular interval, offering significant multiplexing gains. 
We present a detailed mathematical analysis of the beamwidth and beamdepth achievable with the MLA and show that by appropriately selecting the number of antennas in each constituent ULA, ideal near-field beamfocusing can be realized. In addition, we propose a computationally efficient localization method that fuses estimates from each ULA, enabling efficient parametric channel estimation. Simulation results confirm the accuracy of the analytical expressions and that MLAs achieve near-field beamfocusing with a limited number of antennas, making them a promising solution for next-generation wireless systems.
\end{abstract}


\section{Introduction}
\label{S_Intro}

The commercialization of massive multiple-input multiple-output (mMIMO) technology has been a cornerstone of $5$G networks, enabling substantial improvements in spectral efficiency and energy efficiency compared to $4$G \cite{massivemimobook}. As $5$G deployments continue to expand, covering $45\%$ of the world’s population by the end of 2023 \cite{Ericsson2024}, current research focuses on developing better technology to achieve the ambitious goals of $6$G and beyond. Next-generation base stations (BSs) must support even higher cell throughput and new services such as wireless sensing and artificial intelligence (AI) at the edge \cite{2020_Saad_IEEENet,2024_Pourkabirian_Commag}. The BSs will be equipped with MIMO technology that incorporates significantly larger arrays in the upper mid-band, which might enable significantly higher throughput through greatly enhanced spatial multiplexing capabilities \cite{2024_Bjornson_Arxiv}. However, to fully exploit these benefits in practice, it is desirable to utilize the near-field spherical wave property to obtain richer channel features that make closely spaced users easily separable in the joint distance/angular domain  \cite{2023_Ramezani_Bits}. 

\subsection{Related Works}

Three regions classically define the electromagnetic radiation patterns of an antenna array with respect to the propagation distance: the reactive near field, radiative near field, and far field \cite{balanis2016antenna}. We focus on the radiative near field, where amplitude variations across the antennas in the array are negligible, and only phase variations are considered. For simplicity, we refer to the radiative near field as the \emph{near field}.
Unlike conventional far-field beamforming, which focuses signals on a far-away point, near-field beamforming acts like a lens, concentrating signals on a specific location, known as finite-depth beamforming/beamfocusing. This is accomplished using a matched filter (MF) based on the channel coefficient of each antenna in the array.

Near-field beamfocusing makes spatial multiplexing more practically useful, particularly in line-of-sight and sparse multipath environments. The reason is that the array can separate multiple users simultaneously by distinguishing them in the angular and distance domains, instead of only the angular domain, as with traditional far-field beamforming used in legacy networks. The beamfocusing feature exists when the propagation distance is less than the Fraunhofer distance\footnote{{The study in \cite{2025_Sun_CMag} indicates that the classical Fraunhofer distance still provides a reasonably accurate near-field/far-field boundary, with relatively low channel capacity estimation error compared to most alternatives.
}} $2D^2/\lambda$ \cite{2021_Björnson_Asilomar}, where $D$ and $\lambda$ denote the aperture length of the array and the wavelength, respectively. In practice, given a coverage range of $d_{\max}$ and a wavelength of $\lambda$, we can calculate the aperture length $D= \sqrt{d_{\max} \lambda/2}$ required to enable beamfocusing. However, filling this aperture with conventional half-wavelength-spaced antennas is challenging: the array would be physically large, require hundreds or thousands of antennas, and the computational complexity would be excessive.

A natural approach to mitigate these issues is to decrease the number of antennas in the array while maintaining its aperture area. There are two main options to achieve this:
\begin{enumerate}
    \item \textbf{Sparse uniform arrays}: Increasing the antenna spacing in the array beyond half-a-wavelength while maintaining a uniform structure, leading to a linear sparse array (LSA) \cite{2023_Yang_Arxiv,2025_Li_WCMag,Wang2023_GCWKS,2024_Chen_JSTSP}.
    {\cite{2023_Yang_Arxiv} and \cite{2025_Li_WCMag} provide a comprehensive study of how sparse arrays can benefit integrated sensing and communication (ISAC) in the direction of 6G. This includes grating lobe suppression, beam codebook design, and array geometry optimization. In \cite{Wang2023_GCWKS}, the beam-pattern properties of sparse and collocated arrays and the distribution of inter-user angular separations were analyzed, showing their impact on inter-user interference (IUI); the study found that sparse arrays are less likely to suffer severe IUI than collocated arrays. \cite{2024_Chen_JSTSP} focuses on maximizing the sum rate in near-field communications by minimizing grating lobes and IUI.}   
    \item \textbf{Modular arrays}: Creating sub-arrays with the same number of antennas and half-wavelength antenna spacing, while extending the spacing between these sub-arrays to achieve the desired aperture area  \cite{2021_Jeon_Commag}.
    {They propose modular arrays for $6$G, allowing sub-arrays with different shapes and sizes. The paper argues that the required hardware should be feasible by the expected $6$G timeline around 2030 and presents system-level simulations showing the performance gains of modular massive MIMO (mmMIMO).}
\end{enumerate}
Additionally, a combination of these approaches can be implemented, where each sub-array is designed as a sparse array. Examples include coprime arrays \cite{2023_Wang_Arxiv} and extended coprime arrays \cite{2024_Zhou_Arxiv}. The traditional issue with sparse uniform arrays is that they introduce grating lobes \cite{bjornson2024introduction}, i.e., side lobes equally strong as the main lobe. This is undesirable since it causes strong inter-user interference and localization ambiguity. In contrast, modular arrays only create low-gain side-lobes \cite{2024_Li_TWC} since the sub-arrays create ``clean'' beam-directivity.

A typical 5G BS uses a compact array with tens of antennas (e.g., 32 or 64 antennas).
A modular linear array (MLA) can be created by deploying multiple such 5G-like arrays on the same rooftop. This makes synchronization and coordinated processing of these sub-arrays feasible, as they are separated by only a few meters (or less), well within the range of cable installations, and can have a joint local oscillator and baseband unit. Hence, MLAs are easier to deploy and operate than cell-free mMIMO networks, where the sub-arrays are distributed over the coverage area and require local processing capabilities and over-the-air synchronization \cite{cell-free-book}.

Only a few previous studies have been performed of near-field communications using modular arrays \cite{2024_Li_TWC,2024_Meng_WCL,2022_Li_CL,2025_Li_TC,2025_Meng_TC}. \cite{2024_Li_TWC} has proposed a multi-user scheduling method to maximize the sum rate when communicating using an MLA in the near field. In \cite{2024_Meng_WCL}, modular uniform and non-uniform arrays are considered, with a focus on deriving the closed-form Cramér-Rao bound for range and angle estimation in a bistatic near-field sensing system. In \cite{2022_Li_CL},  the performance of modular arrays is analyzed for different array geometries and user locations. The study focuses on analyzing the signal-to-noise ratio (SNR) in modular arrays, where the SNR is analytically derived to increase proportionally with the number of sub-arrays and/or the number of antennas within each sub-array. Codebook design and hybrid beamforming for modular near-field arrays were proposed in \cite{2025_Li_TC} and \cite{2025_Meng_TC}, respectively.

Recently, in terahertz (THz) communications, the concept of modular arrays has also recently emerged \cite{2018_Song_TWC,2022_Yan_TWC,2024_Yang_TWC}. These modular arrays were referred to as widely spaced sub-arrays, first introduced in \cite{2018_Song_TWC},  where it was identified that two types of spatial multiplexing, inter-path and intra-path, can be jointly exploited. Inter-path multiplexing refers to signals observed from individual sub-arrays, which exhibit planar-like wavefronts, while intra-path multiplexing involves signals seen from the entire array, characterized by spherical-like wavefronts. By exploiting both multiplexing types,  an advanced hybrid beamforming architecture tailored for this scenario was developed to provide high spectral efficiency \cite{2022_Yan_TWC}. Further, \cite{2024_Yang_TWC} examined the localization capabilities of such systems, deriving theoretical bounds such as the Cramér-Rao bound for near-field positioning in widely spaced sub-arrays.

To the best of our knowledge, the near-field beamfocusing pattern in MLAs has not been analytically studied in prior work. We note that the beamfocusing pattern in MLA depends on its configuration and  is different from the one in conventional uniform arrays.
Therefore, it is essential  to identify the most efficient MLA configuration (e.g., number of sub-arrays and antennas per sub-array) for achieving effective beamfocusing. Furthermore, given that higher processing complexity and cost arise from handling a larger number of antennas, a critical question---the main motivation of this paper---is whether it is possible to achieve effective beamfocusing with a relatively few antennas.

\subsection{Contributions}

In this paper, we analyze mathematically and numerically how the geometry of MLAs influences their beampattern in near-field communications. We consider key design factors such as the number of antennas in each ULA, the number of ULAs, and the spacing between the ULAs. Based on this analysis, we identify the most efficient geometrical structure for achieving beamfocusing effects with a minimal number of antennas. The main objective of this paper is to show that it is possible to achieve beamfocusing with relatively few antennas. This contrasts with the recent study \cite{2024_Li_TWC}, which focuses mainly on the grating lobes. Additionally, we propose low-complexity localization and channel estimation methods for MLAs, and demonstrate that they achieve near-optimal performance with significantly reduced complexity. The localization approach is based on triangulation, which is known to enable low complexity algorithms, as it operates only in the angular domain rather than jointly in the angular and distance domains \cite{2025_Haghshenas_Arxiv}. {This aligns with the concept of inter- and intra-path multiplexing considered in \cite{2018_Song_TWC,2022_Yan_TWC}.} However, to the best of our knowledge, triangulation has not been applied to develop efficient localization algorithms for modular array systems. More specifically, the novelty of the proposed algorithm lies in the use of the actual modular array geometry, where sub-arrays are naturally in the far-field of the BS. Therefore, far-field angle estimation can be performed in each sub-array. The far-field angles are obtained through array signal processing, specifically, using the MUltiple SIgnal Classification (MUSIC) algorithm. The estimated distances are then fused through a least-squares triangulation step. Thus, the novelty of the approach stems from the overall framework for performing near-field localization, which combines far-field angle estimates with the modular array structure, rather than from the use of the triangulation method itself.

The main contributions are summarized as follows.
\begin{itemize}
    \item We analyze the beampattern achieved by MLAs and develop related analytical expressions, including closed-form expressions for the beamwidth and beamdepth. The analysis begins with two ULAs and then generalizes to an arbitrary number $L$ of ULAs, where $L \geq 2$. The use of MLAs with two sub-arrays shows great potential due to their simple structure, while enabling beamfocusing with significantly fewer antennas compared to the conventional half-wavelength spacing ULA with the same aperture length. 

    \item Using the developed analytical expressions, we calculate the number of antennas required to achieve clean beamfocusing without gain fluctuations inside of the beamfocusing region. We also propose an algorithm to determine the necessary number of sub-arrays to achieve such near-field beamfocusing for a given aperture length. Our novel closed-form expressions for beamwidth and beamdepth enable these calculations and ensure that the algorithm can be executed efficiently.
    
    \item  We develop an efficient user localization method for MLAs by estimating one angle per ULA and combining them to accurately determine the user's location. Consequently, we only need to perform a grid search over the angular dimension, rather than the angular and distance dimensions, as typically required in near-field localization \cite{2023_Chen_Arxiv,2024_Kosasih_Asilomar,2024_Ramezani_Arxiv}. This leads to a significant reduction in complexity while maintaining good localization performance, as will be demonstrated in our simulation results. Additionally, we use the proposed localization method for channel estimation and demonstrate that our approach achieves SE that closely matches that of the one with perfect channel state information.
\end{itemize}
We note that the MLA considered in this paper reduces to an LSA when each ULA has a single antenna and the spacing between the ULAs exceeds $\lambda/2$. The beampattern analysis for this LSA configuration is detailed in \cite{2024_Zhou_Arxiv} and textbooks such as \cite{bjornson2024introduction}. Furthermore, the MLA reduces to a standard ULA (i.e., with half-wavelength spacing) when the spacing between the closest antennas of adjacent ULAs matches the inter-antenna spacing within each ULA. Although beamfocusing pattern analysis for a ULA has not been reported in the prior literature, it can be derived as a special case of beampattern analysis for uniform rectangular arrays, outlined in our recent work \cite{2023_TWC_BD}. We will give this derivation in Lemma~\ref{Fresnel_Approx_ULA}. Hence, we extend the results from \cite{2023_TWC_BD, 2024_Zhou_Arxiv}. Compared to the conference version in \cite{Kosasih2024_Asilomar}, this paper includes three main differences:
\begin{itemize}
\item The conference version considers only the case of $L = 2$ ULAs, whereas this paper extends the analysis to an arbitrary number of ULAs.
\item The conference version does not provide an algorithm for determining the number of arrays required to achieve a desired beamfocusing pattern.
\item The conference version does not address channel estimation or localization for the MLA.
\end{itemize}
These additions lead to substantial improvements in the journal version. In particular, we expand the analysis, introduce new algorithms, and broaden the scope to connect beamfocusing patterns with practical performance metrics in wireless communications, such as spectral efficiency.

\subsection{Outline}

The remainder of the paper is organized as follows. In Section~\ref{Sect_Prelim}, we provide preliminaries for the array gain of an antenna array and beamfocusing using a ULA. In Section~\ref{S_NF_TwoArr}, we analyze the beampattern, including beamwidth and beamdepth, of MLAs with two ULAs. In Section~\ref{S_NF_MultiArr}, we extend the beampattern analysis for MLAs with an arbitrary number of ULAs. In Section~\ref{S_loc_ChEst}, we discuss how localization and channel estimation can be performed by exploiting the structure of MLAs.  This paper is concluded in Section~\ref{Sect_conclude}.
\section{Preliminaries}
\label{Sect_Prelim}

In this section, we elaborate on how to achieve beamfocusing with a ULA receiver and a single antenna user transmitter.\footnote{We take the viewpoint of a receiver array with an isotropic transmitter for ease of presentation, but the same results apply in the reciprocal setup with a transmitting array and isotropic receiver.} The methodology presented in this section will be utilized in later sections when we analyze scenarios with multiple arrays.

We consider a free-space propagation scenario with an isotropic transmit antenna located at $\left(x_t,y_t,z \right)$ and a receiver array deployed along the $x$-axis with its center in the origin. We focus on a broadside transmission where $x_t=y_t=0$, to keep the beamfocusing explanation clear and simple, but non-broadside transmissions will be considered in later sections. 

\subsection{Computing the Array Gain of a ULA}

We consider a ULA with $N$ aperture antennas indexed by $n \in \{1, \dots, N \}$. Antenna $n$ is centered at the point $(\bar{x}_n,0,0)$ given by:
\begin{align}
    \bar{x}_n = \left(n-\frac{N+1}{2}\right)\delta,
\end{align}
where $\delta$ is the spacing between adjacent antennas. Antenna $n$ covers the area $\delta \times \delta$ defined as
\begin{align}
    \mathcal{A}_n = \left\{ (x,y,0) : |x-\bar{x}_n|\leq \frac{\delta}{2}, |y|\leq \frac{\delta}{2} \right\}.
\end{align}
The aperture length of the ULA is defined as $D_{\rm array} =  N \delta $.
The spherical phase variations of the impinging wave across the array are negligible (less than $\pi/8$) when the propagation distance exceeds the Fraunhofer array distance, defined as $d_{\rm FA} =  2D_{\rm array}^2 /\lambda$ \cite{2021_Björnson_Asilomar}.  This corresponds to far-field communication, where the spherical wavefront curvature can be accurately approximated as planar. In contrast, when the propagation distance is shorter than $d_{\rm FA}$, the spherical wavefront curvature becomes noticeable, characterizing the radiative near-field communications. In this paper, we need a channel model that is accurate in both the far and near fields.

{If the transmitter sends a signal polarized in the $y$-dimension,   the electric field at the point $(x,y,0)$ of the receiver aperture can be derived from the Green function as  \cite[App. A]{2020_Björnson_JCommSoc}
\begin{equation}\label{eq_II_ElectField}
 E (x,y) =  \frac{E_0}{\sqrt{4 \pi}} \frac{\sqrt{z \left( x ^2+z^2\right)}} 
 { r ^{5/4}}   e^{-\imagunit\frac{2\pi}{\lambda} \sqrt{r}},
\end{equation}
where $r =  x^2 + y^2 + z^2$ is the squared Euclidean distance between the transmitter and the considered point, and $E_0$ is the electric intensity of the transmitted signal. The expression in \eqref{eq_II_ElectField} takes into account how the effective area, polarization losses, and path loss can vary over the aperture in the radiative near-field, where the detailed descriptions can be found in \cite[Eq. 69]{2020_Björnson_JCommSoc}}. 

The power received at antenna $n$ is given by $E_0^2 |h_n|^2 / \eta$, 
where $E_0^2/\eta$ is the power of the transmitted signal, $\eta$ is the impedance of the free space, 
\begin{equation} \label{eq:channel-n}
h_n = \frac{1}{E_0\sqrt{A}}  \int_{\mathcal{A}_n} E(x,y) dx dy
\end{equation}
is the dimensionless complex-valued channel response \cite[Eq.~(64)]{2020_Björnson_JCommSoc}, and $A=\delta^2$ is the area of each receive antenna.

If $s \in \mathbb{C}$ denotes the information symbol transmitted with power $E_0^2/\eta$, the received signal at antenna $n$ becomes:
\begin{equation}
\chi_{n} = h_{n} s + v_{n}, \quad n=1,\ldots,N,
\end{equation}
where $v_{n} \sim \CN(0,\sigma^2)$ is independent complex Gaussian receiver noise. 
Suppose that linear receive combining is used to detect $s$ from $\chi_{1},\ldots,\chi_{N}$.
Multiplying each received signal by a weight $w_n$ and summing them up, we obtain $\sum_{n=1}^Nw_n\chi_n$. 
Without loss of generality, the weights can be normalized so that $\sum_{n=1}^N|w_n|^2=1$. The SNR of the combined signal is then given as
\begin{align}
    \mathrm{SNR} = \frac{E_0^2}{\eta\sigma^2}\left| \sum_{n=1}^Nw_nh_n\right|^2 = \frac{ \left| \sum_{n=1}^N \int_{\mathcal{A}_n} w_n E(x,y) dx dy\right|^2 }{A \eta\sigma^2}. \label{eq:SNR-not-maximized}
\end{align}
The SNR can be maximized using MF, where $w_{n} = h_{n}^*/\sqrt{\sum_{j=1}^N|h_{j}|^2}$ is the weight selected for antenna $n$. The resulting maximum SNR is\footnote{Note that the MF weights are normalized as 
$\sum_{n=1}^{N} |w_n|^2 = 1$, and this normalization does not 
affect the resulting SNR expression.
}
\begin{equation} \label{eq:SNR-formula}
\overline{\mathrm{SNR}}= 
 \frac{E_0^2}{\eta \sigma^2} \sum_{n=1}^{N}  \left| h_n \right|^2 = \frac{\sum_{n=1}^{{N}} \left| \int_{\mathcal{A}_n} E(x,y) dx dy \right|^2}{A\eta\sigma^2}.
\end{equation}

The weight $w_n$ appears within the integral in \eqref{eq:SNR-not-maximized}. Hence, we can treat the receiving combining as multiplying the electric field by a piecewise constant function $w(x,y)$ that takes the constant value $w_n$ over the antenna $n$. If the antennas are sufficiently small, this combining function becomes nearly continuous, and we can optimize the entire function. The corresponding continuous MF receiver is given by 
\begin{equation} \label{eq:CMF}
    w(x,y)= E^*(x,y) \Big/\sqrt{\sum_{n=1}^N \frac{1}{A} \int_{\mathcal{A}_n}|E(x,y)|^2 dx dy},
\end{equation}
where the normalization ensures that the total average weight over the antennas is one, i.e., $\sum_{n=1}^N \frac{1}{A}\int_{\mathcal{A}_n}|w(x,y)|^2{dxdy}=1 $.
The resulting SNR with continuous MF is:
\begin{align}
\overline{\mathrm{SNR}}&= 
 \frac{E_0^2}{\eta \sigma^2} \sum_{n=1}^{N}\left|h_n\right|^2 \approx  \frac{ \left| \sum_{n=1}^N \int_{\mathcal{A}_n} w(x,y) E(x,y) dx dy\right|^2 }{A \eta\sigma^2} \nonumber \\
 & = 
 \frac{1}{\eta\sigma^2} \sum_{n=1}^N \int_{\mathcal{A}_n}\left|E(x,y)\right|^2dxdy. \label{eq:SNR-continuous}
\end{align}

\begin{remark}
The apparent difference in the placement of the absolute value between 
\eqref{eq:SNR-formula} and \eqref{eq:SNR-continuous} originates from the transition from the discrete MF receiver 
to its continuous-aperture approximation. In \eqref{eq:SNR-formula}, the coherent combining 
is carried out over $N$ discrete antennas, leading to a summation of 
per-antenna integrals whose magnitudes are squared individually. 
In \eqref{eq:SNR-continuous}, the continuous MF combines the field across the entire aperture 
before squaring the magnitude, which is why the absolute value applies 
to the complete integral expression. Both formulations are consistent 
representations of MF combining.
\end{remark}

To quantify the array gain actually achieved using multiple antennas, we consider the ideal reference case with continuous MF and all antennas receiving the same power as the antenna at the origin. In this reference case, the SNR is given as
\begin{align}
  \mathrm{SNR}^{\rm ref} = \frac{N}{\eta\sigma^2}\int_{\mathcal{A}}\left|E(x,y)\right|^2dxdy \label{eq:SNR-cont},
\end{align}
where  $\mathcal{A} = \left\{ (x,y,0) : |x|\leq \frac{\delta}{2}, |y|\leq \frac{\delta}{2}\right\}$ denotes the area covered by a single antenna located at the origin.

We can then write the normalized array gain as {\cite{2021_Björnson_Asilomar}}:
\begin{equation}\label{eq_II_NormalizedArrayGain}
   G_{\rm array} = \frac{ \sum_{n=1}^{{N}} \left| \int_{\mathcal{A}_n} E(x,y) dx dy \right|^2}{N A \int_{\mathcal{A}} \left|E(x,y)\right|^2 dx dy },
\end{equation}
obtained by dividing the SNR achieved by conventional MF in \eqref{eq:SNR-formula} by the SNR achieved in the reference continuous MF case in \eqref{eq:SNR-cont}. Here, $\mathcal{A}_n$ denotes the set of points corresponding to the physical aperture area of antenna $n$. The value of $G_{\rm array}$ ranges from $0$ to $1$.

Whenever the distance to the transmitter is larger than $2D_{\rm array}$, the electric field expression in \eqref{eq_II_ElectField} can be simplified using the Fresnel approximation as \cite{2021_Björnson_Asilomar}
\begin{equation}\label{eq:ElectricalField_Fresnel_approx}
    E(x,y) \approx \frac{E_0}{\sqrt{4 \pi} z} e^{-\imagunit \frac{2\pi}{\lambda}\left(z + \frac{x^2}{2z} + \frac{y^2}{2z} \right)}.
\end{equation}
This expression is obtained by omitting amplitude variations over the aperture of the array and making a first-order Taylor expansion of the Euclidean distance between the transmitter and the array in the phase expression. In the remainder of the paper we will utilize the Fresnel approximation in \eqref{eq:ElectricalField_Fresnel_approx} for analytical derivations. 
The tightness of this approximation was recently demonstrated in \cite{2023_TWC_BD}. 

\subsection{Beamfocusing Using a ULA}
\label{subsec:beamfocusing-ULA}

MF focuses the reception/transmission on a desired focal point where the maximum achievable array gain should be achieved. The array gain will also be significant at other nearby points. 
In traditional far-field beamforming, the maximum array gain is achieved at any other point in the same direction, regardless of the distance.
Beamfocusing refers to the alternative near-field phenomenon in which the array gain is only large in a limited distance range around the focal point.
Next, we will describe how this phenomenon can be quantified.

Suppose that the MF is configured for reception from a focal point $(0,0,F)$ while the transmitter is located at $(0, 0, z)$. Using the Fresnel approximation, the MF weight designed for antenna $n$ is
\begin{align}
    w_n= \frac{\left(\int_{\mathcal{A}_n} e^{-\imagunit \frac{2\pi}{\lambda}\left(\frac{x^2}{2F}+\frac{y^2}{2F}\right)} d x d y\right)^*}{\sqrt{\sum_{j=1}^N  \left\vert \int_{\mathcal{A}_j}e^{-\imagunit \frac{2\pi}{\lambda}\left(\frac{x^2}{2F}+\frac{y^2}{2F}\right)} d x d y     \right \vert^2   }}.
    \label{eq:combining-coefficient}
\end{align}
The normalized array gain expression can then be obtained following the same principles as in \eqref{eq_II_NormalizedArrayGain}. Hence, we divide the SNR in \eqref{eq:SNR-not-maximized} by the reference SNR in \eqref{eq:SNR-cont} to obtain
\begin{align}\label{eq:MF_near}
&{G}_{\rm ULA}  = \frac{\left\vert  \sum_{n=1}^N  \int_{\mathcal{A}_n } w_n E(x,y) dx dy   \right\vert^2 }{N A \int_{\mathcal{A}} \left|E(x,y)\right|^2 dx dy }
\nonumber\\
& \approx \hat{G}_{\rm ULA} = \frac{1}{NA^2}\left\vert \sum_{n=1}^N  \int_{\mathcal{A}_n}  w(x,y)   e^{-\imagunit \frac{2\pi}{\lambda}\left(\frac{x^2}{2z}+\frac{y^2}{2z}\right)} dx dy \right\vert^2
\nonumber\\
&=  \frac{1}{(NA)^2}  
\left\vert \sum_{n=1}^N \int_{\mathcal{A}_n} e^{+\imagunit\frac{2\pi}{\lambda}\left(\frac{x^2}{2F}+\frac{y^2}{2F}\right)}e^{-\imagunit\frac{2\pi}{\lambda}\left(\frac{x^2}{2z}+\frac{y^2}{2z}\right)}   dx dy\right \vert^2,
\end{align}
where the corresponding continuous MF receiver in \eqref{eq:CMF} has the combining function
\begin{equation}
    w(x,y) = e^{\imagunit \frac{2\pi}{\lambda}\left(\frac{x^2}{2F}+\frac{y^2}{2F}\right)}/\sqrt{N}
\end{equation}
since $\sum_{n=1}^N \frac{1}{A} \int_{\mathcal{A}_n}|E(x,y)|^2 dx dy=N A /A = N$.
The approximation in \eqref{eq:MF_near} follows from replacing MF with continuous MF (i.e., assuming small antennas) and using the Fresnel approximation in \eqref{eq:ElectricalField_Fresnel_approx}. We can derive a closed-form expression for the normalized array gain in \eqref{eq:MF_near}, as stated in the following lemma. 
\begin{lemma}\label{Fresnel_Approx_ULA}
When the transmitter is located at $(0,0,z)$ and the MF is focused on $(0, 0, F)$, the Fresnel approximation-based normalized array gain for a ULA becomes
\begin{equation}\label{eq_III_ApproxGainRect}
\hat{G}_{\rm ULA} =
\frac{\left( {{C}^{2}}\left( \sqrt{ a } \right)+{{S}^{2}}\left( \sqrt{ a } \right) \right) \left( {{C}^{2}}\left( \sqrt{ a } N \right)+{{S}^{2}}\left( \sqrt{ a } N \right) \right)}{( N a )^{2}},
\end{equation}
 where $C(\cdot )$ and $S(\cdot )$ are the Fresnel integral functions \cite{1956_Polk_TAP}, $ a = \frac{\lambda}{8{z}_{\rm eff}}$, and $z_{\rm eff} = \frac{Fz}{|F-z|} $.
\end{lemma}
\begin{proof}
The proof is given in Appendix~\ref{App_Fresnel_Approx}
and is inspired by the derivation in \cite[Eq. (22)]{1956_Polk_TAP}.
\end{proof}

\section{MLA via Two Separated ULAs on the Same Line}
\label{S_NF_TwoArr}

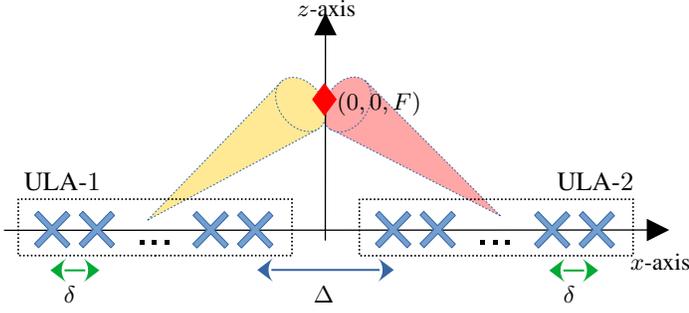
\begin{figure}
    \centering
     \begin{overpic}[width=\linewidth]{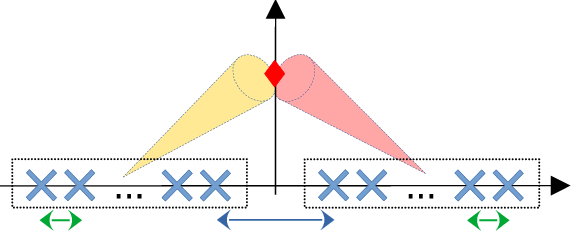}
     \put(3,14){ULA-$1$}
     \put(83,14){ULA-$2$}
     \put(84,-3){\small $ \delta$}
     \put(9,-3){\small $ \delta$}
     \put(46.5,-3){\small $\Delta$}
     \put(50,26){\small $(0,0,F)$}
     \put(44,40){\small$z$-axis}
     \put(94,2){\small$x$-axis}
     \end{overpic}
    \vspace{0.7mm}
    \caption{MLA composed by two ULAs deployed $\Delta$\,m apart. Both ULAs are focused at $(0,0,F)$.}
    \label{F_twoULAs}
    \vspace{-3mm}
\end{figure}

We will now consider an isotropic transmitter located at $(0,0,z)$, communicating with a receiver, equipped with an MLA consisting of two ULAs. The centers of the innermost antennas in the two ULAs are separated by $\Delta$ meters, as illustrated in Fig.~\ref{F_twoULAs}.  
We note that the two ULAs are positioned along the same line on the $x$-axis, which differs from the typical cell-free massive MIMO setup where ULAs can be deployed at arbitrary locations.
Each of the ULAs is equipped with $N$ antennas. The $n$-th antenna element for ULA-$\ell \in \{1,2\}$ is located at the point
\begin{equation}\label{eq:coord_twoULas}
    \bar{x}_n^{(\ell)} = \left( n - \frac{N+1}{2}\right) \delta + {\left( \ell - \frac{3}{2}\right) \left(\Delta + (N-1) \delta\right)}.
\end{equation} 
Notice that the second term in \eqref{eq:coord_twoULas} is either plus or minus $\overline{\Delta}$, where $\overline{\Delta} = \frac{\Delta + (N-1) \delta}{2}$.
A special case of the above configuration occurs when $\Delta= \delta$, in which case the two ULAs will form a single ULA with inter-element spacing $ \delta$. The aperture length of the array is $D_{\rm array} = \Delta + (2N-1)\delta $. The Fraunhofer distance of the array increases with the separation $\Delta$ between the two ULAs. This might imply that beamfocusing can be achieved even with a small total number of antennas deployed in the two ULAs (small aperture length in each ULA) by appropriately adjusting the separation $\Delta$. To verify this, we will analytically study the beamfocusing behavior for this MLA model.

We begin by adapting the normalized array gain with the MF-based beamfocusing expression in \eqref{eq:MF_near} to account for the two ULAs in the MLA. The resulting expression is given in \eqref{eq:G_M2_ULA}, at the top of the next page,
\begin{figure*}
\begin{align}
& G_{2} = \frac{1}{2NA}
\frac{\Bigg\vert \sum\limits_{n=1}^N \int_{\mathcal{A}_{1,n}}     w_{1,n} E(x,y) dx dy   +\sum\limits_{n=1}^N \int_{\mathcal{A}_{2,n}}  w_{2,n} E(x,y) dx dy  
\Bigg\vert^2}{\int_{\mathcal{A}} \left|E(x,y)\right|^2 dx dy}.
\label{eq:G_M2_ULA}
\end{align}
\end{figure*}
where $\mathcal{A}_{\ell,n}$ is the square area covered by the antenna-$n$ of ULA-$\ell$. By utilizing the Fresnel approximation, we obtain the normalized array gain expression:
\begin{multline}\label{eq:G2app_A}
\hat{G}_2  = \frac{1}{2NA^2}\Bigg\vert \sum_{n=1}^N \int_{\mathcal{A}_{1,n}}  w_{1,n}   e^{-\imagunit \frac{2\pi}{\lambda}\left(\frac{(x-\Delta)^2}{2z}+\frac{y^2}{2z}\right)} dx dy + 
\\
\sum_{n=1}^N \int_{\mathcal{A}_{2,n}} w_{2,n}    e^{-\imagunit \frac{2\pi}{\lambda}\left(\frac{(x+\Delta)^2}{2z}+\frac{y^2}{2z}\right)} dx dy 
\Bigg\vert^2,
\end{multline} 
where the MF weights $w_{1,n}$ and $w_{2,n}$ are given as in \eqref{eq:combining-coefficient2a}-\eqref{eq:combining-coefficient2b} at the top of the next page. Following the same methodology as in Section~\ref{subsec:beamfocusing-ULA}, they can be approximated by the continuous MF functions  $w_{1}(x,y) = e^{+\imagunit\frac{2\pi}{\lambda}\left(\frac{(x-\Delta)^2}{2F}+\frac{y^2}{2F}\right)}/\sqrt{2N}$ and  $w_{2}(x,y) = e^{+\imagunit\frac{2\pi}{\lambda}\left(\frac{(x+\Delta)^2}{2F}+\frac{y^2}{2F}\right)}/\sqrt{2N}$.
\begin{figure*}
\begin{align}
 &   w_{1,n}= \frac{\left(\int_{\mathcal{A}_{1,n}} e^{-\imagunit \frac{2\pi}{\lambda}\left(\frac{(x-\Delta)^2}{2F}+\frac{y^2}{2F}\right)} d x d y\right)^*}{\sqrt{\sum\limits_{j=1}^{N} \left( \left\vert \int_{\mathcal{A}_{1,j}}e^{-\imagunit \frac{2\pi}{\lambda}\left(\frac{(x-\Delta)^2}{2F}+\frac{y^2}{2F}\right)} d x d y     \right \vert^2 + \left\vert \int_{\mathcal{A}_{2,j}}e^{-\imagunit \frac{2\pi}{\lambda}\left(\frac{(x+\Delta)^2}{2F}+\frac{y^2}{2F}\right)} d x d y     \right \vert^2 \right)  }},\label{eq:combining-coefficient2a} \\
&    w_{2,n}= \frac{\left(\int_{\mathcal{A}_{2,n}} e^{-\imagunit \frac{2\pi}{\lambda}\left(\frac{(x+\Delta)^2}{2F}+\frac{y^2}{2F}\right)} d x d y\right)^*}{\sqrt{\sum\limits_{j=1}^{N} \left( \left\vert \int_{\mathcal{A}_{1,j}}e^{-\imagunit \frac{2\pi}{\lambda}\left(\frac{(x-\Delta)^2}{2F}+\frac{y^2}{2F}\right)} d x d y     \right \vert^2 + \left\vert \int_{\mathcal{A}_{2,j}}e^{-\imagunit \frac{2\pi}{\lambda}\left(\frac{(x+\Delta)^2}{2F}+\frac{y^2}{2F}\right)} d x d y     \right \vert^2 \right)  }}.\label{eq:combining-coefficient2b}
\end{align}
\hrulefill
\end{figure*}
Therefore, we can approximate \eqref{eq:G2app_A} by 
 \begin{align}
\notag& \hat{G}_2 = \frac{1}{{{(2  N A)}^{2}}}  \cdot \\ \notag
& \Bigg|  
\int\limits_{-\frac{N \delta}{2} }^{\frac{N\delta}{2} }
\int\limits_{-\frac{ \delta}{2}}^{\frac{ \delta}{2}} 
e^{\imagunit \frac{2\pi}{\lambda} \left(\frac{(x-\overline{\Delta})^2}{2F} + \frac{y^2}{2F}\right)}
{e}^{-\imagunit\frac{2\pi }{\lambda }\left(\frac{{{(x-\overline{\Delta} )}^{2}}}{2z}+\frac{{{y}^{2}}}{2z} \right)} dxdy  
\\ 
& +  \int\limits_{-\frac{N\delta}{2} }^{\frac{N\delta}{2} }\int\limits_{-\frac{ \delta}{2}}^{\frac{ \delta}{2}} 
e^{\imagunit \frac{2\pi}{\lambda} \left(\frac{(x+\overline{\Delta})^2}{2F} + \frac{y^2}{2F}\right)}
{e}^{-\imagunit\frac{2\pi }{\lambda }\left(\frac{{(x + \overline{\Delta} )^{2}}}{2z}+\frac{{{y}^{2}}}{2z} \right)}  dxdy \Bigg|^2 \label{eq:G2app}.
\end{align}

\begin{figure}
\centering
\subfloat[A single ULA.]
{\includegraphics[width=0.508\textwidth]{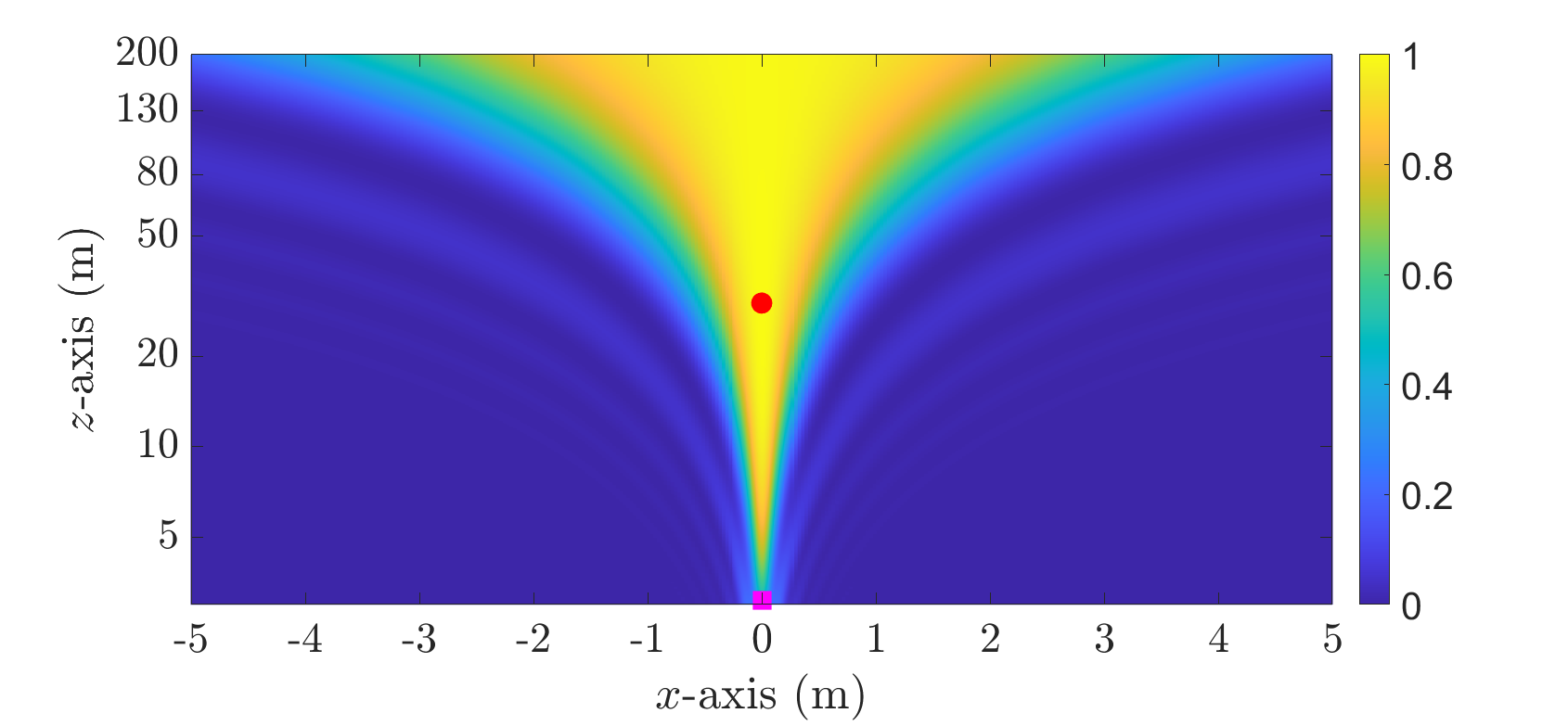}}\hfill
\centering
\subfloat[Two ULAs separated by $5$\,m apart.]
{\includegraphics[width=0.51\textwidth]{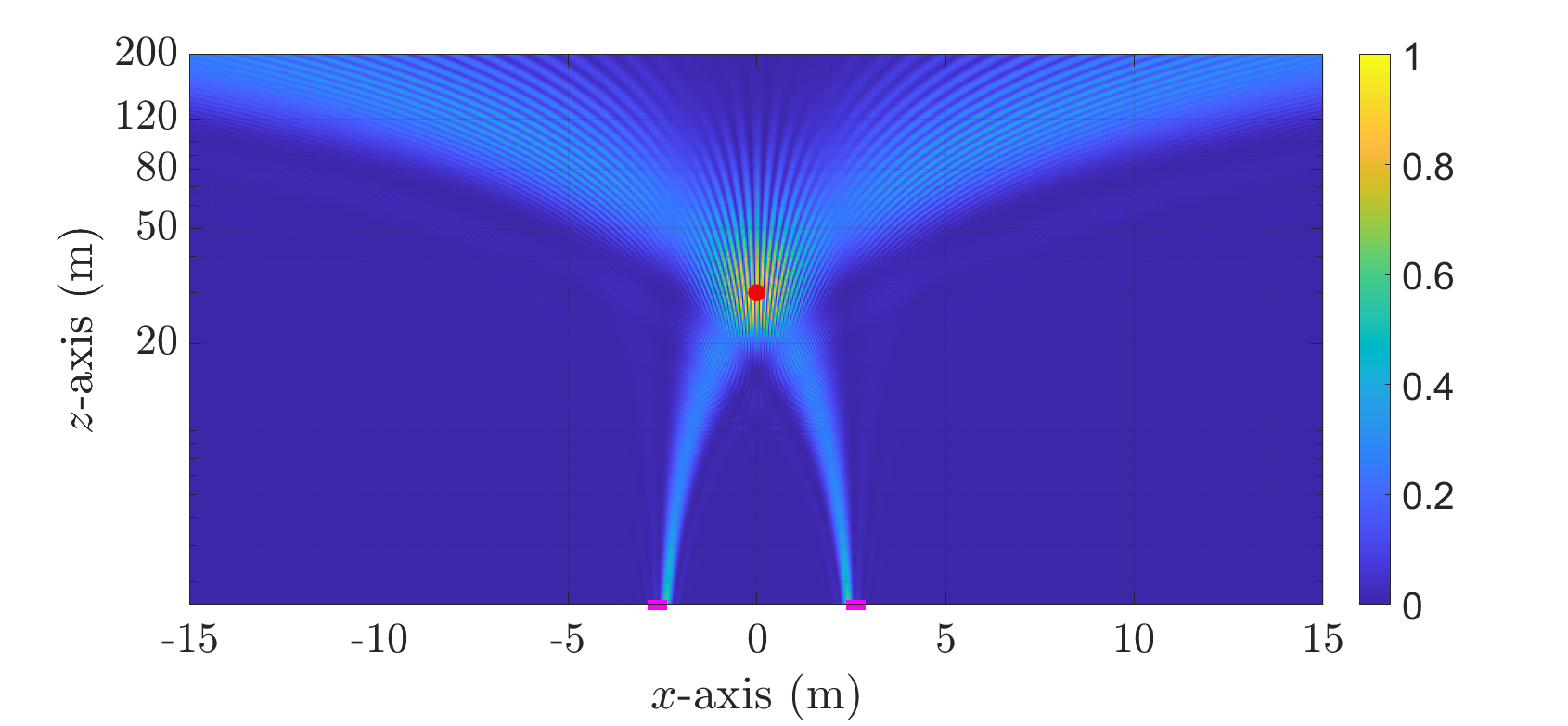}}
\caption{The beampattern of a single ULA and MLA. The coloring shows the normalized array gain. The magenta rectangles indicate the location of the ULA, while the red circle indicates the location of the user. We use a logarithmic scale on the $z$-axis.}
\label{F_BW}
\end{figure}

Let us compare the normalized array gain of a standard ULA with that of an MLA. The standard ULA has \( N = 50 \) half-wavelength-spaced antennas, while the MLA consists of two ULAs, each with \( 25 \)  half-wavelength-spaced antennas and separated by a distance of \( \Delta = 5 \) meters.  We set the focus of the array at the location $(0,0,30)$ and the arrays operate at the $15$\,GHz\footnote{Throughout the paper, we consider a carrier frequency of $15$\,GHz, which is a promising upper mid-band frequency band for the next wireless generation \cite{2024_Bjornson_Arxiv}. However, the conclusions presented apply to other carrier frequencies if the simulation setup is scaled accordingly.} carrier frequency. 
Fig.~\ref{F_BW}(a) depicts the normalized array gain (beampattern) of the ULA across the $xz$-plane. The array gain is larger in the yellow and green regions. We observe that the beam energy spreads out infinitely behind the user. This is a typical beampattern for far-field beamforming and was expected since the user distance is greater than the Fraunhofer array distance of $23$\,m. In contrast, Fig.~\ref{F_BW}(b) shows the MLA beampattern, where the signal is focused at the same location. The array gain is only large in a limited region around the focal point. This shows that beamfocusing can be achieved using the MLA configuration, even with the same total number of antennas as the ULA. In the following, we will provide a theoretical characterization of the beampattern from the MLA.

 \subsection{Beamwidth Analysis}\label{S_Beamwidth_2arrays}

The transverse beamwidth (BW) around the focal point $(0,0,F)$  can be analyzed by considering the normalized array gain obtained along the  $x$-axis.
Specifically, we compute the array gain that a signal transmitted from another point $(x_t,0,F)$ with coordinate $x_t$ along the $x$-axis achieves. Note that the distance $F$ along the $z$-axis is the same as for the focal point. The new coordinate along the $x$-axis must be incorporated into the phase-shift expression when computing the normalized array gain.
Generalizing \eqref{eq:G2app} to arbitrary $x_t$ and setting $z=F$ , we can then write the normalized array gain approximation in \eqref{eq:BW_twoArr_first} at the top of the next page. 
\begin{figure*}
\begin{align}
\notag
    &\hat{G}_{2,x_t}  = \frac{\left |
      \left(
       \int\limits_{-\frac{N\delta}{2} }^{\frac{N\delta}{2} }
      \int\limits_{-\frac{ \delta}{2}}^{\frac{ \delta}{2}} 
       e^{\imagunit \frac{2\pi}{\lambda} \left(\frac{(x-\overline{\Delta})^2}{2F} + \frac{y^2}{2F}\right)}
     {e}^{-\imagunit\frac{2\pi }{\lambda }\left(\frac{{{(x-\overline{\Delta} -x_t)}^{2}}}{2F}+\frac{{{y}^{2}}}{2F} \right)} dx
    dy  +  \int\limits_{-\frac{N\delta}{2} }^{\frac{N\delta}{2} }\int\limits_{-\frac{ \delta}{2}}^{\frac{ \delta}{2}} 
     e^{\imagunit \frac{2\pi}{\lambda} \left(\frac{(x+\overline{\Delta})^2}{2F} + \frac{y^2}{2F}\right)}
     {e}^{-\imagunit\frac{2\pi }{\lambda }\left(\frac{{(x + \overline{\Delta} -x_t)^{2}}}{2F}+\frac{{{y}^{2}}}{2F} \right)}  dxdy \right) \right|^{2}}{{{(2  N \delta^2)}^{2}}}  \nonumber 
    \\
    &=\frac{1}{{{(2  N \delta^2)}^{2}}}  \left|
      \delta  \left(\int\limits_{-\frac{N\delta}{2}  -\overline{\Delta}}^{\frac{N\delta}{2}  -\overline{\Delta}}
       e^{\imagunit \frac{2\pi}{\lambda F} u_1 x_t} du_1 + 
     \int\limits_{-\frac{N\delta}{2}   +\overline{\Delta}}^{\frac{N\delta}{2}   +\overline{\Delta}} e^{\imagunit \frac{2\pi}{\lambda F} u_2 x_t} du_2
      \right) \right|^{2} 
    =\Bigg|\frac{1}{2} \sinc\left(\frac{N x_t}{2F}\right) \left(e^{\imagunit \left(\frac{2\pi \overline{\Delta} x_t}{\lambda F}\right) } +   e^{-\imagunit \left(\frac{2\pi \overline{\Delta} x_t}{\lambda F}\right) } \right)  \Bigg|^2 
    \label{eq:BW_twoArr_first}
\end{align}
\hrulefill
\end{figure*}
The final expression can be simplified as
\begin{equation} \label{eq:BW_twoArr}
    \hat{G}_{2,x_t} = \sinc^2 \left(\frac{N x_t}{2F}\right) \cos^2\left(\frac{2 \pi \overline{\Delta} x_t}{\lambda F} \right)    \leq \sinc^2\left(\frac{N x_t}{2F}\right),
\end{equation}
where the upper bound is the envelope of the function since the cosine-term oscillates more rapidly than the sinc-term.

\begin{figure}
\centering
\subfloat[Total aperture length of $D_{\rm array} = 1$\,m.  ]
{ \includegraphics[width=0.45\textwidth]{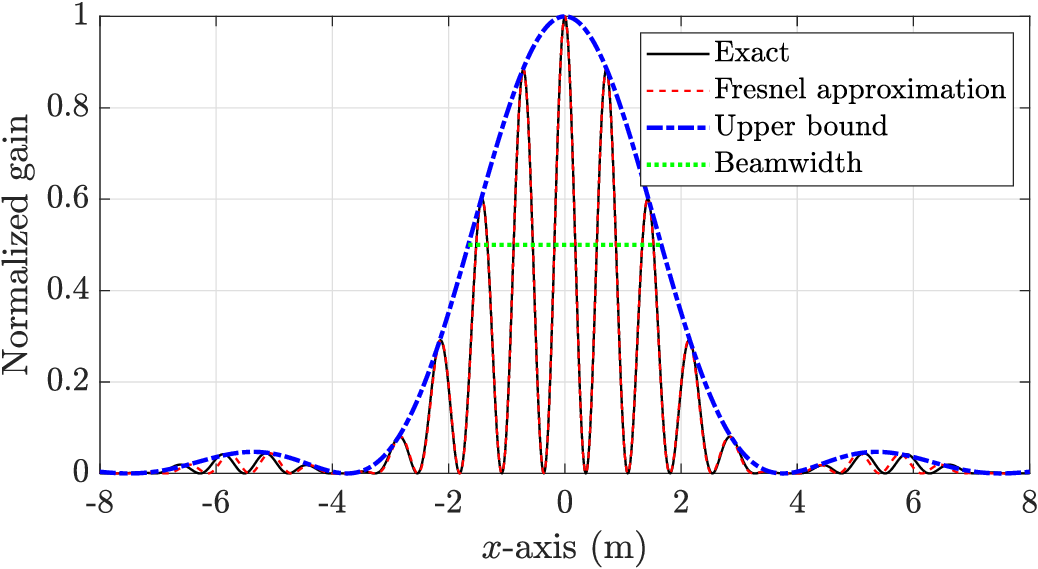}}\hfill
\centering
\subfloat[Total aperture length of $D_{\rm array} = 2$\,m.]
{ \includegraphics[width=0.45\textwidth]{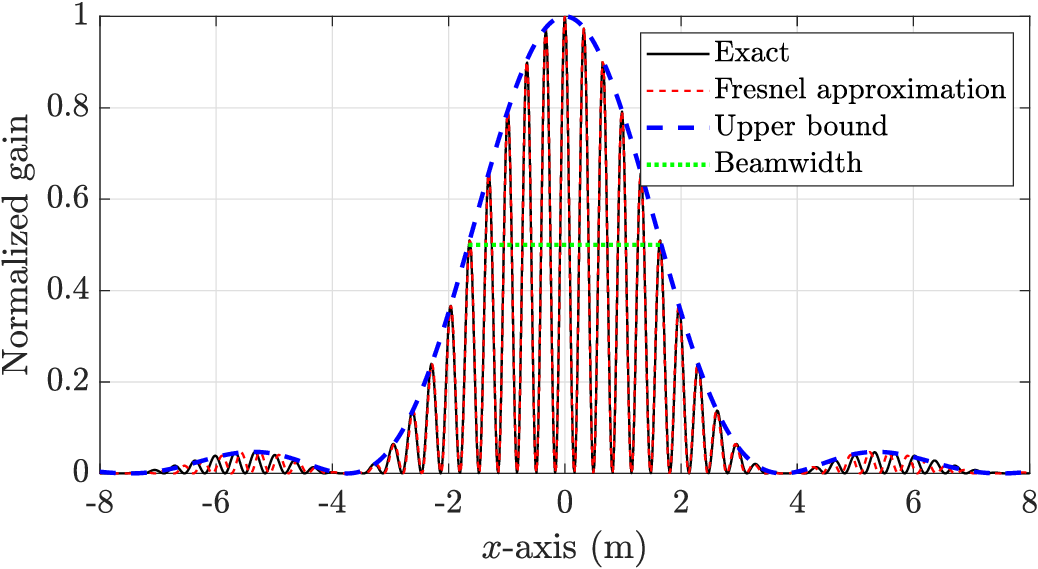}}
\caption{The beampattern of MLA along the $x$-axis in a setup with two ULAs, each having $N=16$ antennas.}
\label{F_BW2}
\end{figure}

We will use the term \emph{beamfocusing region} to refer to the area where the envelope of the normalized array gain is larger than $0.5$. This is in line with the classical half-power beamwidth concept. The half-power beamwidth can be analytically expressed as: 
\begin{equation}\label{eq:3dB_BW}
    0.443 \approx \frac{N |x_t|}{2F} \rightarrow {\rm BW_{\rm 3 dB}} \approx \frac{1.77F}{N}, 
\end{equation}
since  $\sinc^2(0.443)\approx 0.5$. Note that the upper bound is independent of the inter-array spacing $\Delta$.  However, the spacing will determine the oscillations in the actual gain pattern, as illustrated by the following example.

We evaluate the beamwidth of the MLA in Fig.~\ref{F_BW2}. Each ULA in the MLA is equipped with $N=16$ antennas. The focal point of the MF is at $(0,0,30)$\,m and we show the normalized array gain for points $(x_t,0,30)$\,m for different values of $x_t$. The black curve is the exact normalized array gain (computed without approximation), and the red-dashed curve is based on the analytical formula for $\hat{G}_{2,x_t}$ in \eqref{eq:BW_twoArr}, which was computed using the Fresnel approximation and the continuous MF approximation. The blue-dashed curve corresponds to the upper bound (i.e., envelope) in \eqref{eq:BW_twoArr} and the green-dashed line shows the interval where the upper bound is above 0.5, which we refer to as the (half-power) beamwidth. 

In Fig.~\ref{F_BW2}(a), the aperture length of the MLA is set to $D_{\rm array} =  1$\,m. We observe that the normalized array gain oscillates along the $x$-axis. The gain oscillates more when we consider the larger aperture length of $D_{\rm array} =  2$\,m in Fig.~\ref{F_BW2}(b) (i.e., a larger distance between the ULAs). Gain fluctuations are undesirable because they imply that the array gain can drop to zero unexpectedly if a user located at the focal point moves within the beamfocusing region. Therefore, we will now analytically derive conditions for eliminating them. We first characterize the number of peaks that the ripple has within the beamfocusing region. The peaks appear when $\cos^2(\frac{2 \pi \overline{\Delta}x_t}{\lambda F}) = 1$. Since $\cos^2\left(  \pi k\right)=1$ for any integer value of $k$, the peaks occur when 
\begin{equation}
    x_t = \pm\frac{\lambda F}{2 \overline{\Delta}} k.
\end{equation}
There is a peak in the center of the beamfocusing region and since the distance between two adjacent peaks is $\frac{\lambda F}{2\overline{\Delta}}$, the number of peaks within the beamfocusing region becomes 
\begin{equation}\label{eq:num_nulls}
     2\left\lfloor\frac{{\rm BW_{3dB}}/2} {\lambda F/ (2  \overline{\Delta})}  \right\rfloor +1 \approx 2\left\lfloor\frac{1.77 \overline{\Delta}} {N \lambda}  \right\rfloor +1.
\end{equation}
This expression reveals that the number of ripples within the beamfocusing region decreases when the number of antennas per ULA $(N)$ increases or when $\overline{\Delta}$ decreases. We only have one peak within the beamfocusing region  if $\frac{1.77 \overline{\Delta}} {N \lambda}<1$. Therefore,  the number of antennas in each ULA that satisfies this condition is
\begin{equation}\label{eq:Nmin}
    N > {\frac{ 1.77 \overline{\Delta}}{\lambda}}.
\end{equation}
Substituting $\overline{\Delta} =  \frac{\Delta + (N-1) \delta}{2}$ into \eqref{eq:Nmin}, we  obtain
\begin{equation}\label{eq:Nmin2a}
    1.13 > {\frac{{\Delta}}{N \lambda}} + \frac{(N-1)}{N \lambda} \delta.
\end{equation}
If $\delta = \lambda/2$, we obtain the following condition:
\begin{equation}\label{eq:Nmin2}
    0.63 + \frac{1}{2N} >  \frac{\Delta}{N \lambda}.
\end{equation}
Since the total aperture length of the MLA with two separated ULAs is $D_{\rm array} = (2N-1)\frac{\lambda}{2} + \Delta$, we can rewrite \eqref{eq:Nmin2} to express the condition in terms of  the aperture length of the two ULAs and the total aperture length of the MLA:
\begin{equation}\label{eq:Nmin3}
    \frac{N \lambda}{D_{\rm array}}  \geq  0.62.
\end{equation}
Hence, the total aperture length of the individual ULAs must be at least $62 \%$ of the total aperture length of the MLA to obtain only a single focused beam.

To illustrate this result, we plot the beampattern across the $xz$-plane and the normalized array gain along the $x$-axis in Fig.~\ref{F_Beampattern_nice}, when the MF is focused on $(0,0,30)$. We set the total aperture length of the MLA to $D_{\rm array} = 2$\,m and the number of antennas per ULA to $N=64$. Hence, the separation between the array is $\Delta = 0.72$\,m and the ratio $(N \lambda/D_{\rm array})$ is  $64 \%$ satisfying the condition in \eqref{eq:Nmin3}. Fig.~\ref{F_Beampattern_nice} shows that there is only one peak within the beamfocusing region (i.e., above the green line). This contrasts to Fig.~\ref{F_BW2}(b), where the total aperture length was the same, but there are fewer antennas per ULA, so the beamfocusing region is wider and has $9$ peaks inside it. An alternative way to achieve ripple-free beamfocusing is to fill the entire $2$\,m aperture length with half-wavelength-spaced antennas. This would require $200$ antennas since $\lambda=0.02$\,m in this example. In contrast, the MLA configuration only required $128$ antennas to achieve the desired beamfocusing capability.

\begin{figure}
\centering
\subfloat[Beampattern in $xz$-plane.]
{\includegraphics[width=0.9\linewidth]{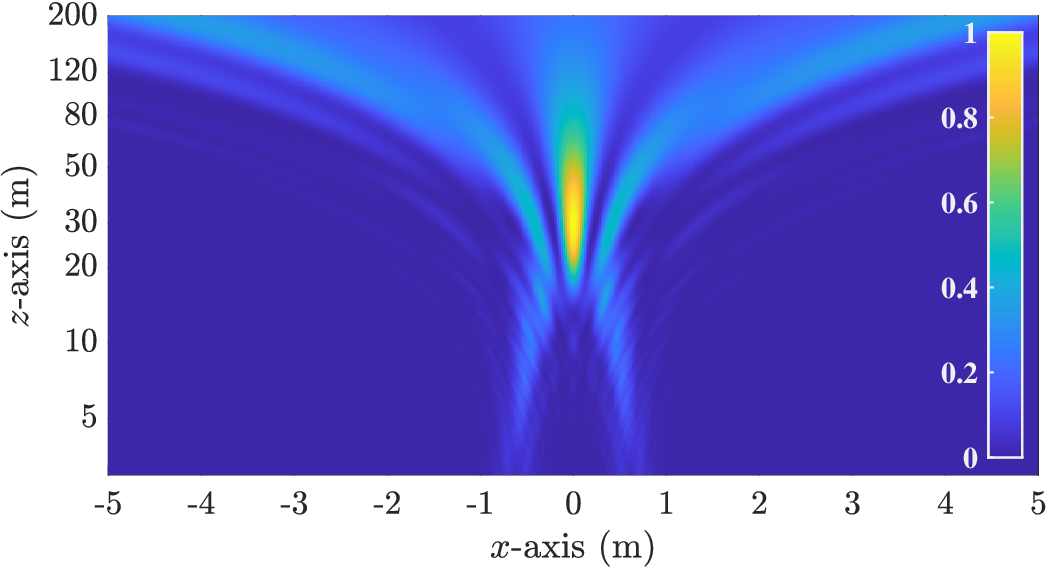}}\hfill
\centering
\subfloat[Normalized array gain at $(x,0,30)$.]
{\includegraphics[width=0.9\linewidth]{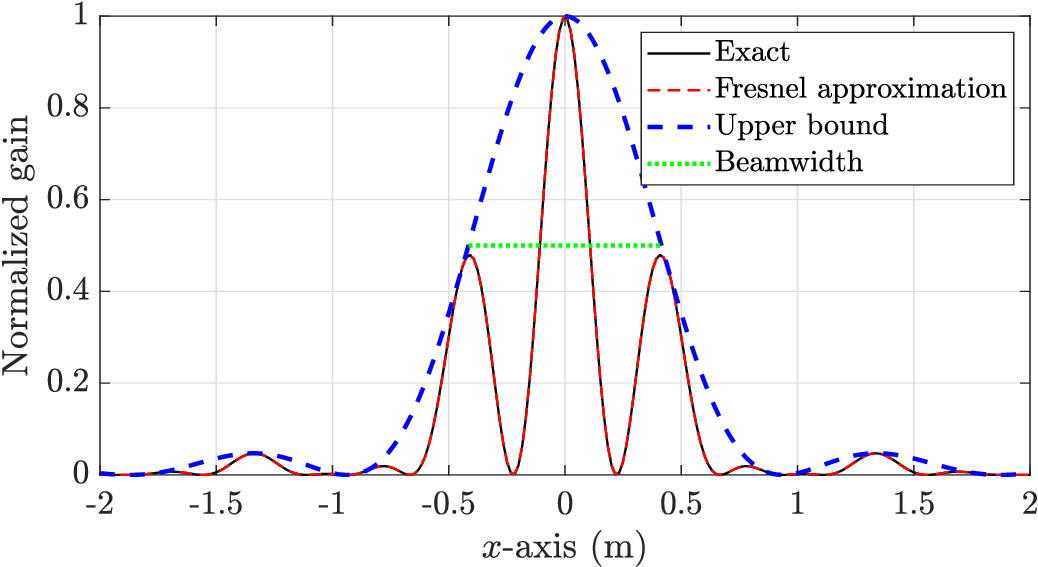}}
\caption{The normalized array gain achieved by beamfocusing from an MLA with two ULAs, each with $N=64$ antennas, with the focal point $(0,0,30)$\,m.}
\label{F_Beampattern_nice}
\end{figure}
 
\subsection{Beamdepth Analysis}

The depth of the beam can also be analyzed using the normalized array gain. When beamfocusing is performed, we expect the normalized array gain to be only above $0.5$ in a finite interval along the $z$-axis. To analyze this, we assume that the MF is focused at $(0,0,F)$ and consider potential transmitter locations expressed as $(0,0,z_t)$. We can compute the following closed-form expression for how the normalized array gain in \eqref{eq:G2app} varies with $z_t$.
  
\begin{theorem}\label{Fresnel_Approx_twoULAs}
Suppose that the transmitter is located at $(0,0,z_t)$ and that MLA uses MF focused on $(0, 0, F)$.
The  Fresnel approximation of the normalized array gain is
\begin{multline}\label{eq_III_ApproxGainTwoULAs}
\hat{G}_{2,z_t} =  \frac{1}{(2 N a)^2}  \left( C^2(\sqrt{a}) + S^2(\sqrt{a}) \right)  \\
\cdot\left( \left( C(\beta_1) +C(\beta_2) \right)^2 + \left( S(\beta_1) +S(\beta_2) \right)^2\right),
\end{multline}
 where   $\beta_1 = \sqrt{a} N + \sqrt{\frac{2 }{\lambda z_{\rm eff}}}\overline{\Delta}$,  $\beta_2 = \sqrt{a} N - \sqrt{\frac{2 }{\lambda z_{\rm eff}}}\overline{\Delta}$, $a = \frac{\lambda}{8{z}_{\rm eff}}$, and $z_{\rm eff} = \frac{Fz_t}{|F-z_t|} $.
\end{theorem}
\begin{proof}
The proof is given in Appendix~\ref{App_Fresnel_Approx_twoULAs}.
\end{proof}

\begin{figure}
    \centering
    \includegraphics[width=0.9\linewidth]{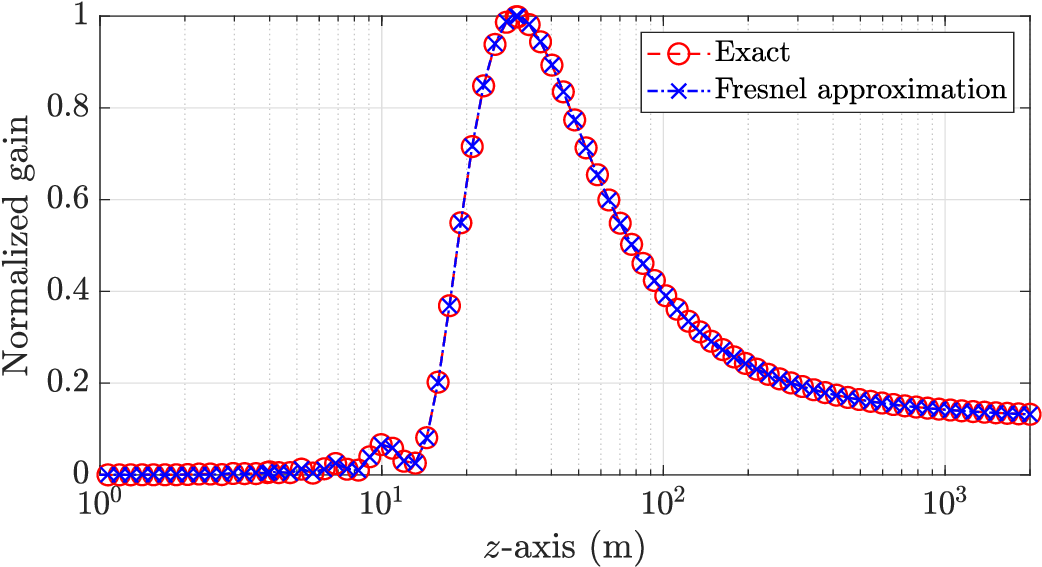}
    \caption{The Fresnel approximation of MLA with the number of antennas in each ULA $N=64$ and focal point $(0,0,30)$\,m.}
    \label{F_BD1}
\end{figure}
 We evaluate the Fresnel approximation of the normalized array gain in \eqref{eq_III_ApproxGainTwoULAs} by plotting Fig.~\ref{F_BD1}. The setup in Fig.~\ref{F_BD1} is identical to the one in Fig.~\ref{F_Beampattern_nice}. The Fresnel approximation results in virtually the same gain as the exact value, obtained by numerically evaluating the array gain expression in \eqref{eq:G_M2_ULA} with integrated MF at focal point $(0,0,30)$\,m. 
\section{MLA via Multiple Separated ULAs}
\label{S_NF_MultiArr}

We now generalize the beamwidth and beamdepth analysis from the previous section to a setup where the MLA consists of $L \geq 2$ ULAs. For notational convenience, we assume that $L$ is an even number, but similar results can be established for odd values. The $L$ ULAs are located along the same line on the $x$-axis and each consists of $N$ antennas. The distance between the centers of the closest antennas to any adjacent ULA is $\Delta$ meters.  Hence, antenna $n$ in ULA $\ell \in \{1,...,L\}$ is located at the $x$-coordinate
\begin{equation}\label{eq:coord_multiULas}
    \bar{x}_n^{(\ell)} = \left( n - \frac{N+1}{2}\right) \delta + {\left( \ell - \frac{L+1}{2}\right) \left(\Delta + (N-1) \delta\right)}.
\end{equation}
For any impinging electric field $E(x,y)$, receive combining weights $w_{\ell,n}$,  the normalized antenna array gain for antenna $n$ in ULA $\ell$  is
\begin{align}
& G_{\rm MLA} = 
\frac{\Bigg\vert \sum_{\ell=1}^L \sum_{n=1}^N \int_{\mathcal{A}_{\ell,n}}  w_{\ell,n}   E(x,y) dx dy   
\Bigg\vert^2}{LNA\int_{\mathcal{A}} \left|E(x,y)\right|^2 dx dy}.
\label{eq_NormalizedArrayGain_multi}
\end{align}
If an isotropic transmitter is located at $(x_t,0,z)$, then we can use the same Fresnel approximation and small antenna approximation as in the previous section. We can then obtain the approximate normalized array gain in \eqref{eq_G_MLA}, at the top of the next page, where $w(x,y)$ denotes the continuous receiver filter.
\begin{figure*}
\begin{align}
\hat{G}_{\rm MLA}  = \frac{\left| \sum_{\substack{k=1 \\ k \text{ odd}}}^{L-1}\left(
\int\limits_{-\frac{N\delta}{2} }^{\frac{N\delta}{2} }
\int\limits_{-\frac{ \delta}{2}}^{\frac{ \delta}{2}} 
w(x,y)
{e}^{-\imagunit\frac{2\pi }{\lambda }\left(\frac{{{(x-k\overline{\Delta} -x_t)}^{2}}}{2z}+\frac{{{y}^{2}}}{2z} \right)} dx
dy  +  \int\limits_{-\frac{N\delta}{2} }^{\frac{N\delta}{2} }\int\limits_{-\frac{ \delta}{2}}^{\frac{ \delta}{2}} 
w(x,y)
{e}^{-\imagunit\frac{2\pi }{\lambda }\left(\frac{{(x + k\overline{\Delta} -x_t)^{2}}}{2z}+\frac{{{y}^{2}}}{2z} \right)}  dxdy \right) \right|^{2}}{LNA^2}.
\label{eq_G_MLA}
\end{align}
\end{figure*}
Note that equation \eqref{eq_NormalizedArrayGain_multi} is a generalization of \eqref{eq:G_M2_ULA} and the starting point for deriving the beamwidth and beamdepth expressions in this section.

\subsection{Beamwidth Analysis}

We will now analyze the beamwidth when MF is used with the focal point $(0,0,F)$. We consider a transmitting user located at $(x_t,0,z=F)$, where $x_t$ is a variable. With a continuous MF, the normalized gain expression in \eqref{eq_G_MLA} becomes \eqref{eq:3dB_BW_multi_first} at the top of the next page.
\begin{figure*}
\begin{align}
\notag
&\frac{\left| \sum_{\substack{k=1 \\ k \text{ odd}}}^{L-1}\left(
\int\limits_{-\frac{N\delta}{2} }^{\frac{N\delta}{2} }
\int\limits_{-\frac{ \delta}{2}}^{\frac{ \delta}{2}} 
e^{\imagunit \frac{2\pi}{\lambda} \left(\frac{(x-k\overline{\Delta})^2}{2F} + \frac{y^2}{2F}\right)}
{e}^{-\imagunit\frac{2\pi }{\lambda }\left(\frac{{{(x-k\overline{\Delta} -x_t)}^{2}}}{2F}+\frac{{{y}^{2}}}{2F} \right)} dx
dy  +  \int\limits_{-\frac{N\delta}{2} }^{\frac{N\delta}{2} }\int\limits_{-\frac{ \delta}{2}}^{\frac{ \delta}{2}} 
e^{\imagunit \frac{2\pi}{\lambda} \left(\frac{(x+k\overline{\Delta})^2}{2F} + \frac{y^2}{2F}\right)}
{e}^{-\imagunit\frac{2\pi }{\lambda }\left(\frac{{(x + k\overline{\Delta} -x_t)^{2}}}{2F}+\frac{{{y}^{2}}}{2F} \right)}  dxdy \right) \right|^{2}}{(LN\delta^2)^2}\nonumber \\
&=\frac{\left|
\delta \sum_{\substack{k=1 \\ k \text{ odd}}}^{L-1} \left(\int\limits_{-\frac{N\delta}{2}  -k\overline{\Delta}}^{\frac{N\delta}{2}  -k\overline{\Delta}}
e^{\imagunit \frac{2\pi}{\lambda F} u_1 x_t} du_1 + 
\int\limits_{-\frac{N\delta}{2}   +k\overline{\Delta}}^{\frac{N\delta}{2}   +k\overline{\Delta}} e^{\imagunit \frac{2\pi}{\lambda F} u_2 x_t} du_2
\right) \right|^{2}}{{{(L  N \delta^2)}^{2}}} = \left|\frac{1}{L} \sinc\left(\frac{N x_t}{2F}\right) \sum_{\substack{k=1 \\ k \text{ odd}}}^{L-1}  \left(e^{\imagunit \left(\frac{2\pi k \overline{\Delta} x_t}{\lambda F}\right) } +   e^{-\imagunit \left(\frac{2\pi  k \overline{\Delta} x_t}{\lambda F}\right) } \right)  \right|^2
\label{eq:3dB_BW_multi_first}
\end{align}
\hrulefill
\end{figure*}
The last expression can be simplified as
\begin{align}
    \sinc^2\left(\frac{N x_t}{2F}\right) \Bigg| \frac{2}{L} \sum_{\substack{k=1 \\ k \text{ odd}}}^{L-1} \cos\left(\frac{2k \pi  \overline{\Delta} x_t}{\lambda F} \right)  \Bigg|^2   \leq \sinc^2\left(\frac{N x_t}{2F}\right),\label{eq:3dB_BW_multi}
\end{align}
where the upper bound can be viewed as the envelope of the expression since the cosine terms oscillate more rapidly with $x_t$. This expression is a generalization of \eqref{eq:BW_twoArr}, where we consider $L=2$. Interestingly, the envelope is independent of $L$ and represents the array gain variation of a single ULA with $N$ antennas. Therefore, the largest possible beamwidth of the MLA remains the same regardless of the number of ULAs that it consists of, as long as $N,x_t,$ and $F$ are the same. The ULA configurations only determine the ripples.

As discussed earlier, an ideal beamfocusing pattern should only have one peak in the transverse dimension with a normalized array gain above $0.5$. To find the number of arrays that gives us this desired beamfocusing effect, we propose Algorithm~\ref{Alg_L}. In principle, the algorithm works by increasing the number of arrays by $2$ until only a peak exists within the $3$\,dB BW. The $x$-axis within the $3$\,dB BW is discretized into the $R_{\rm BW}$ grids, which we set to $300$. The function $\text{CountPeaks}(\cdot)$ first interpolates the input data and then counts the number of peaks in the resulting curve.

\begin{algorithm}
    \caption{Calculate the required number of arrays $(L)$}
    \begin{algorithmic}[1]
        \State \textbf{Input:} Total aperture length $D_{\rm array}$, focal distance $F$, number of antennas per ULA $N$, wavelength $\lambda$, and grid resolution of the bandwidth $R_{\rm BW}$ 
        \State \textbf{Output:} Number of arrays $L$
        \State \textbf{Initialization:} $t_0 \gets \infty$, $L \gets 0$
        \State Calculate ${\rm BW}_{\rm 3dB} = \frac{1.77F}{N}$
        \State Discretize the range $-{\rm BW}_{\rm 3dB}/2$ to ${\rm BW}_{\rm 3dB}/2$ into $R_{\rm BW}$ points: $x_t^{(1)}, \dots, x_t^{(R_{\rm BW})}$
        \While{$t_0 > 1$ \textbf{and} $L N \delta < D_{\rm array}$}
            \State Increment $L$ by 2: $L \gets L + 2$
            \State Calculate $\Delta = \frac{D_{\rm array} - (L (N-1) + 1)  \delta}{L-1}$  
            \State Calculate $\overline{\Delta} = \frac{\Delta + (N-1) \delta}{2}$
            \For{$i \gets 1$ \textbf{to} $R_{\rm BW}$}
                \State Calculate $\hat{G}^{(i)}_{\rm BW}$ based on \eqref{eq:3dB_BW_multi}:
                \[
                \hat{G}^{(i)}_{\rm BW} =\sinc^2\left(\frac{N x_t^{(i)}}{2F}\right) \left|  \frac{2}{L}\sum_{\substack{k=1 \\ k \text{ odd}}}^{L-1} \cos\left(\frac{2k \pi \overline{\Delta} x_t^{(i)}}{\lambda F} \right) \right|^2
                \]
            \EndFor
            \State $t_0 \gets \text{CountPeaks}\left(\hat{G}_{\rm BW}^{(1)},\cdots,\hat{G}_{\rm BW}^{(R_{\rm BW})}\right)$
        \EndWhile
        \State \textbf{return} $L$
    \end{algorithmic}\label{Alg_L}
\end{algorithm}
To give a concrete example, Fig.~\ref{F_NvsL} shows the number of arrays $L$ required to achieve the desired beamfocusing effect with respect to the number of antennas $N$ per ULA. We consider an MLA with $D_{\rm array} = 2$m and $F=30$\,m.  We observe that the more antennas per ULA, the fewer arrays required.
This makes sense because having more antennas in each ULA reduces the beamwidth of the envelope and reduces the gap between adjacent ULAs. 
For example, with $N=64 $, only $ L=2 $ arrays are needed, as supported by our analysis in Section~\ref{S_Beamwidth_2arrays} (see also Fig.~\ref{F_Beampattern_nice}). In contrast, when there is only one antenna per ULA, $ L=62 $ is required to achieve the desired beamfocusing effect. This configuration resembles an LS---a ULA with a spacing greater than $ \lambda/2 $.

The configuration that employs a minimum number of antennas in the entire MLA ($NL$) is the LSA. However, an LSA will suffer from grating lobes \cite{2023_Yang_Arxiv,bjornson2024introduction}, which means that it will inadvertently create multiple focal points where the maximum array gain is achieved. This issue is avoided by using an MLA, which also has ripples, but sinc envelopes show that the maximum array gain is only achieved at the desired location. Moreover, when additional objectives such as localization and channel estimation performance are considered, we will demonstrate in Section~\ref{S_loc_ChEst} that an MLA with multiple antennas in each ULA provides excellent performance. We note that there exists a complexity-performance trade-off, with the optimal balance varying depending on the application. 

\begin{figure}
    \centering
    \includegraphics[width=0.9\linewidth]{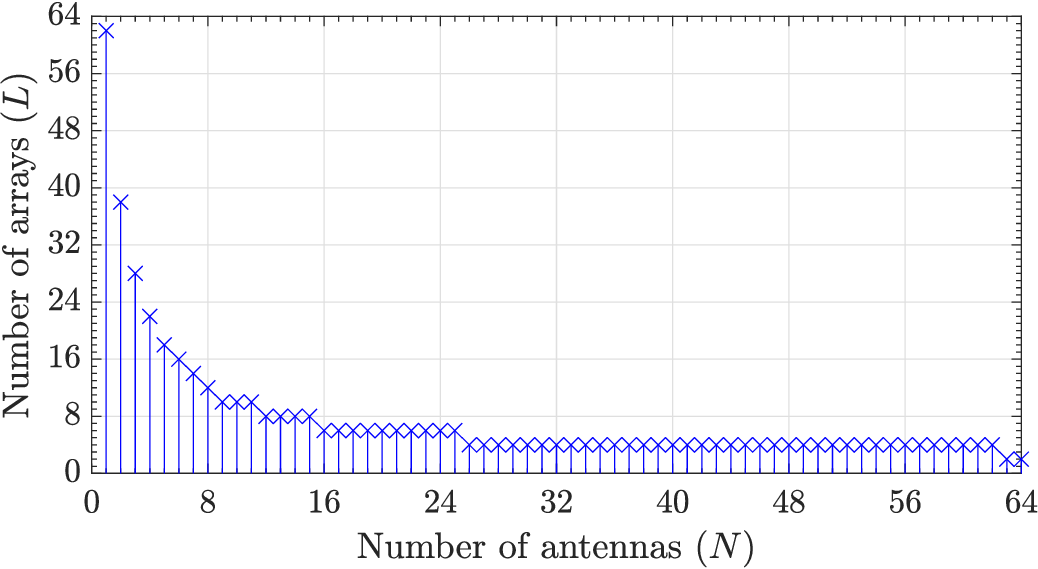}
    \caption{The required number of ULAs for achieving the beamfocusing effect with respect to the number of antennas per ULA. We consider $D_{\rm array} = 2$\,m and $F=30$\,m.}
    \label{F_NvsL}
\end{figure}

\subsection{Beamdepth Analysis}

We will now analyze the beamdepth by considering a potential transmitter at $(0, 0, z_t)$ while the MF is focused at the point $(0, 0, F)$.
We can derive the closed-form expression of the normalized array gain expression in \eqref{eq_G_MLA}, stated in the following corollary, which generalizes Theorem~\ref{Fresnel_Approx_twoULAs}.

\begin{corollary}\label{Fresnel_Approx_multiULAs}
Suppose that the transmitter is located at $(0,0,z_t)$ and that the MLA with $L$ ULAs uses MF focused on $(0, 0, F)$.
The  Fresnel approximation of the normalized array gain is
\begin{align}
\notag
& \hat{G}_{L,z_t} = 
\frac{1}{(L N a)^2}  \left( C^2(\sqrt{a}) + S^2(\sqrt{a}) \right)  \times \\ 
&\hspace{-2mm} \left( \left(\sum_{\substack{k=1 \\ k \text{ odd}}}^{L-1} \left(C(\beta_1^k) +C(\beta_2^k)\right) \right)^2 + \left(\sum_{\substack{k=1 \\ k \text{ odd}}}^{L-1} \left(S(\beta_1^k) +S(\beta_2^k)\right) \right)^2\right),
{\label{eq:Fresnel_approx_multi}}
\end{align}
 where  $\beta_1^k = \sqrt{a} N + \sqrt{\frac{2 }{\lambda z_{\rm eff}}} k \overline{\Delta}$,  $\beta_2^k = \sqrt{a} N - \sqrt{\frac{2 }{\lambda z_{\rm eff}}} k \overline{\Delta}$,  $a = \frac{\lambda}{8{z}_{\rm eff}}$, and $z_{\rm eff} = \frac{Fz_t}{|F-z_t|} $.
\end{corollary}
\begin{proof}
    The proof follows by extending the summation in Appendix~\ref{App_Fresnel_Approx_twoULAs} to $L$ ULAs and proceeding with the same steps.
\end{proof}

\begin{figure}
    \centering
    \includegraphics[width=0.9\linewidth]{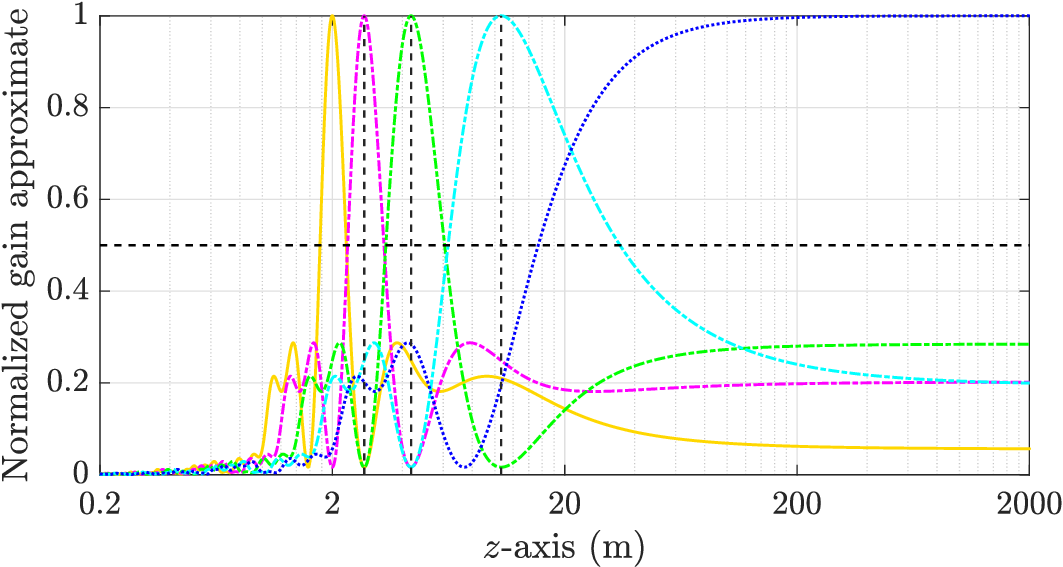}
    \caption{Near-field beamfocusing with a MLA with $L=4$ ULAs, each equipped with $N=16$ antennas. The total array aperture length is set to $D_{\rm array} = 1$\,m. Each color indicates a different focusing location, with the focus adjusted towards the null of the adjacent beam.}
    \label{F_multiplex}
\end{figure}

The finite-depth beamfocusing is a highly desirable feature in near-field propagation scenarios since it allows us to perform spatial multiplexing in the distance domain. To demonstrate this feature when using an MLA, we consider a multi-user scenario where the single-antenna users are located in the same angular direction but at different distances. 
Fig.~\ref{F_multiplex} shows the normalized array gain when focusing at different distances using different colors. We consider an MLA with $L=4$ ULAs, each of which has $N=16$ antennas.
First, we plot the yellow curve corresponding to the focal point $F = 2D_{\text{array}} = 2$\,m. Next, we plot the magenta curve, setting the focus at $F = 2.74$\,m (indicated by the dashed black vertical line), which is approximately at the first null of the yellow curve. Similarly, we plot the green and blue curves, setting their focal points to the first nulls of their respective preceding curves. Note that the closed-form expression in \eqref{eq:Fresnel_approx_multi} enables efficient computation of the curves in Fig.~\ref{F_multiplex} by selecting a particular focal point $F$ and then evaluating the function for a range of $z_t$ values in the distance domain.
\section{Localization and Channel Estimation}
\label{S_loc_ChEst}

The near-field beampattern depends on both the angle and distance, as demonstrated in the previous sections. This property can be utilized for user localization and parametric channel estimation, but at the expense of significantly increased computational complexity. This is because additional search across the distance domain is required in the near field, in contrast to the far field where only angular search is needed. In this section, we propose a low-complexity localization method for MLA-based systems that allows us to achieve a high positioning accuracy with a complexity that depends only on the angular parameter. The localization results are then used to estimate the channel between the MLA and the user, based on the near-field parametric channel model.

\subsection{User Localization}

Suppose that an MLA is used to localize a user in the $xz$-plane. In Fig.~\ref{F_Loc_illus}, we illustrate a user localization scenario with an MLA consisting of four ULAs, each of which is equipped with $N$ antennas. The aperture of an MLA is typically much larger than that of each constituent ULA. Consequently, the Fraunhofer distance of the MLA is significantly greater than that of the individual ULAs. Consider, for example, an MLA with a $2$\,m aperture operating at a $15$\,GHz carrier. Suppose it comprises four ULAs, each with $N=36$ half-wavelength–spaced elements. In this setup, the Fraunhofer distance of the MLA is $400$\,m, whereas that of each ULA is only $12.96$\,m. Therefore, the user can be assumed to be in the near-field of the MLA but in the far-field of each individual ULA. In this section, we present an efficient localization method specifically designed for this operating regime.

Suppose that the user transmits the (possibly random) signal $u[\tau]$ at time $\tau$. The received signal at ULA $\ell$ and time $\tau$ is 
\begin{equation}\label{eq:rec_signal_perULA}
    \vect{y}^{(\ell)}[\tau] = \sqrt{P \beta}\vect{a}\left(\varphi^{(\ell)}\right) u[\tau] + \vect{n}^{(\ell)}[\tau], \quad \tau = 1,\cdots,T,
\end{equation}
where $\vect{a}\left(\varphi^{(\ell)}\right) = \left[e^{\imagunit \frac{2\pi}{\lambda} \bar{x}_1^{(\ell)}  \cos\left(\varphi^{(\ell)}\right)}, \cdots,   e^{\imagunit \frac{2\pi}{\lambda} \bar{x}_N^{(\ell)}  \cos\left(\varphi^{(\ell)}\right)} \right]^{\Ttran}$ denotes the classical far-field array response vector\footnote{{The array response vector expression is valid only when the user is in the far-field of each ULA. If the user also lies in the near-field of the sub-arrays, a different modeling approach would be required.}}, with $\bar{x}_{n}^{(\ell)}$ defined in \eqref{eq:coord_multiULas}. Here,  $\beta = \left(\frac{\lambda}{4 \pi d}\right)^2$ is the large-scale channel gain obtained from the Friis free space path-loss, $d$ denotes the distance between the user and the center of the MLA,  $P$ denotes the transmit power, and $\vect{n}^{(\ell)}[\tau]\in \mathbb{C}^{N}  $ is the noise vector with each element following an independent circular-symmetric complex Gaussian distribution with variance $\sigma^2$.

\begin{figure}
    \centering
     \begin{overpic}[width=\linewidth]{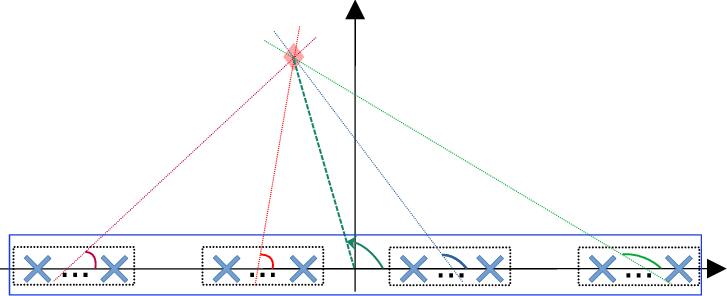}
     \put(37.5,35){\small user}
     \put(28,25){\small $d^{(1)}$}
     \put(13,6.5){\small $\varphi^{(1)}$}
     \put(55,25){\small $d^{(4)}$}
     \put(87,6.5){\small $\varphi^{(4)}$}
     \put(44,20){\small $d_t$}
     \put(50,6.5){\small $\varphi_t$}
     \put(2,10){\small ULA-$1$}
     \put(28,10){\small ULA-$2$}
     \put(53,10){\small ULA-$3$}
     \put(79,10){\small ULA-$4$}
     \put(42,-3){\small The MLA}
     \put(45,42){\scriptsize $z$-axis}
     \put(97,0){\scriptsize $x$-axis}
     \end{overpic}
    \vspace{0.7mm}
    \caption{Localization using MLA composed by $4$ ULAs.}
    \label{F_Loc_illus}
    \vspace{-3mm}
\end{figure}
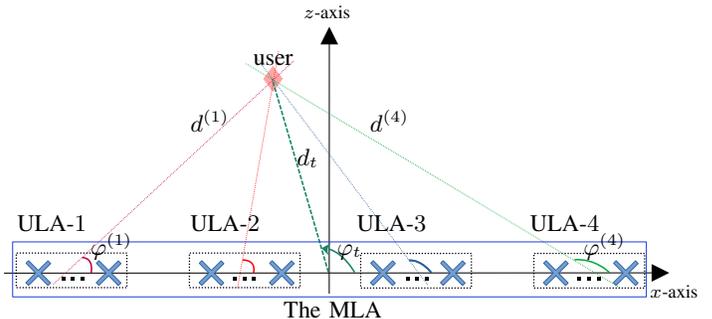

We use the MUSIC algorithm\footnote{Although we consider a single-user scenario, extending it to a multi-user case is straightforward because the MUSIC algorithm inherently supports the detection of multiple users' directions.} to estimate the angle between each individual ULA and the user. The MUSIC algorithm works by exploiting the structure of the eigenvectors of the sample covariance matrix:
\begin{align}
\widehat{\vect{R}} = \frac{1}{T} \sum_{t=1}^T \vect{y}^{(\ell)}[\tau](\vect{y}^{(\ell)}[\tau])^{\Htran}\label{Rl}
\end{align}
over an interval of $T$ received signals. Notice that the received signal should not be a scalar, and therefore, an MLA with $N>1$ is required. Since we consider only a single user, we construct the noise-subspace matrix $\widehat{\vect{U}}_{\rm n}\in \mathbb{C}^{N \times (N-1)}$ whose columns are the eigenvectors of $\widehat{\vect{R}}$ corresponding to the smallest $(N-1)$ eigenvalues. The one-dimensional ($1$D)-MUSIC spectrum across the angular domain is computed as \cite{Stoica2005_book}:
\begin{align} \label{eq:MUSIC-spectrum}
S_M(\varphi)=\frac{1}{\vect{a}^{\Htran}(\varphi)\widehat{\vect{U}}_{\rm n}\widehat{\vect{U}}_{\rm n}^{\Htran}\vect{a}(\varphi)}.
\end{align}
The angle estimate of ULA $\ell$, $\hat{\varphi}^{(\ell)}$, is obtained by identifying the peak in the MUSIC spectrum. 

The angular estimates $\hat{\varphi}^{(1)},\cdots,\hat{\varphi}^{(L)}$ of all ULAs can then be combined to obtain the exact position of the user. We propose using the geometric intersection-based least squares estimator, which determines the user’s location by finding the least squares solution to a set of directional lines. These lines are derived from angle-of-arrival (AoA) measurements at known ULA positions:
\begin{equation}
    \begin{bmatrix}
\hat{x}_t \\
\hat{z}_t
\end{bmatrix}
=
\left(
\begin{bmatrix}
 \tan(\hat{\varphi}^{(1)}) & -1  \\
\vdots & \vdots \\
\tan(\hat{\varphi}^{(\ell)}) & -1  \\
\vdots & \vdots \\
 \tan(\hat{\varphi}^{(L)}) &  -1 
\end{bmatrix}
\right)^{\dagger}
\begin{bmatrix}
\bar{x}^{(1)}\tan(\hat{\varphi}^{(1)}) \\
\vdots \\
\bar{x}^{(\ell)}\tan(\hat{\varphi}^{(\ell)}) \\
\vdots \\
\bar{x}^{(L)}\tan(\hat{\varphi}^{(L)})
\end{bmatrix}, \label{eq:LS}
\end{equation}
where $^{\dagger}$ denotes the pseudoinverse, $\hat{x}_t$ and $\hat{z}_t$ denote the coordinate estimates of the user in the $xz$-plane, and $\bar{x}^{(\ell)}$ indicates the position of the center of the ULA $\ell$ along the $x$-axis. The distance and angle estimates between the user and the center of the MLA (i.e., the origin) can be obtained by using the Cartesian to polar coordinate transformation, denoted as $\hat{d}_t$ and $\hat{\varphi}_t$, respectively. Note that the least-squares approach in \eqref{eq:LS} may produce an ill-conditioned matrix, which can make the inversion unreliable. This corresponds to cases where the estimated directional lines are inaccurate, causing the intersection point to be highly imprecise or even impossible to determine. Such situations typically occur when noise dominates the received signal, i.e., at very low SNR. This is less likely in near-field communication settings. Nevertheless, a regularization term can be added to the least-squares formulation to avoid unstable solutions and prevent the estimate from becoming unacceptable.

Note that the proposed localization algorithm is based on an unweighted least-squares estimator  in \eqref{eq:LS}, implicitly assuming that all sub-arrays provide AoA estimates of identical reliability. In practice, however, heterogeneity among sub-arrays may arise due to differences in aperture size, SNR, or hardware impairments. A natural extension of our method would therefore be to adopt a weighted least-squares formulation, where the weights are selected according to the estimation variance or SNR of each sub-array.

\begin{figure}
\centering
\subfloat[The number of arrays is $L=4$.]
{\includegraphics[width=0.9\linewidth]{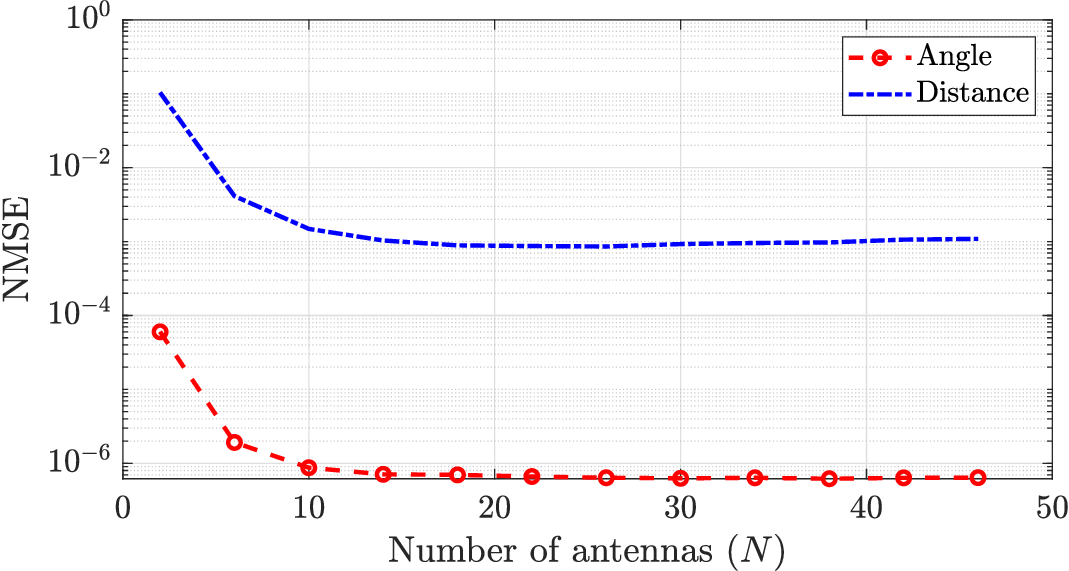}}\hfill
\centering
\subfloat[The number of antennas is $N=16$.]
{\includegraphics[width=0.9\linewidth]{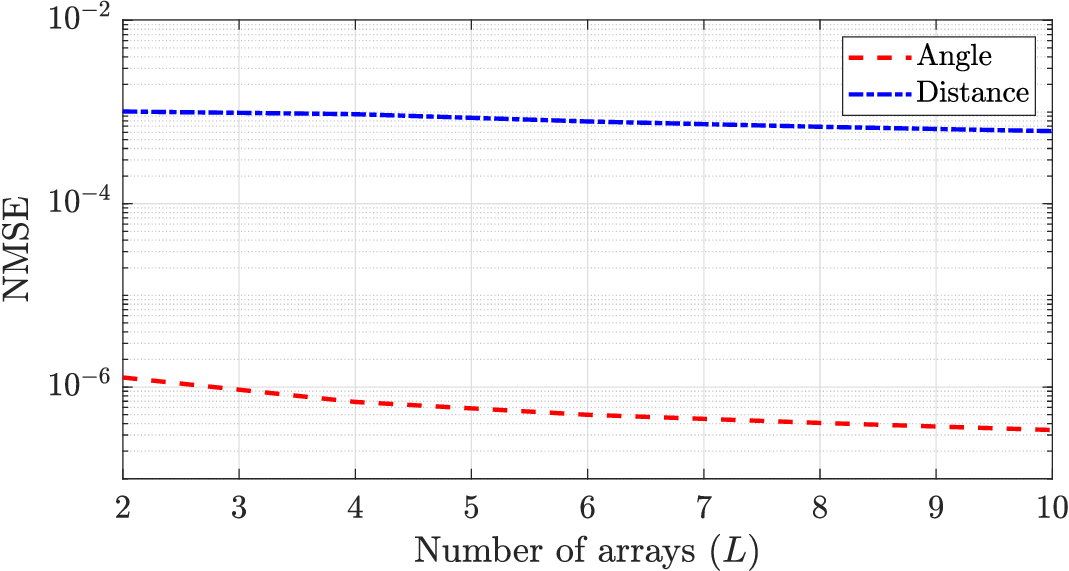}}
\caption{Performance evaluation of the proposed localization method based on varying numbers of antennas and multiple arrays. }
\label{F_locs}
\end{figure}

We will evaluate the performance of the proposed localization approach in terms of the normalized mean square error (NMSE) defined as 
\begin{equation}
    {\rm NMSE }= \frac{ \mathbb{E} \{ \|\hat{\vect{x}} - \vect{x} \|^2 \} }{ \mathbb{E}\{\|  \vect{x}\|^2\}},
\end{equation}
where $\hat{\vect{x}}$ is the estimate of $\vect{x}$.

In Fig.~\ref{F_locs}, we set the aperture length of the MLA to $2$\,m, the number of samples is $T=100$, the transmit power is $P = 20$\,dBm, and the bandwidth is $400$\,MHz. The noise power is $\sigma^2 = -78$\,dBm, including a noise figure of $10$\,dB. We consider a BS array located at the origin, with a user randomly positioned at an angle $\varphi_t \sim \mathcal{U}(-60^\circ, 60^\circ)$ and a distance $d_t \sim \mathcal{U}(4\,\text{m}, 40\,\text{m})$.
 The angular domain is discretized with a resolution of $0.002$ radians.

As depicted in Fig.~\ref{F_locs}(a), NMSE performance improves when a larger number of antennas is considered.  NMSE saturates for $N> 20$. The saturation is due to the limited grid resolutions. If the resolution is increased, the NMSE performances for both the angle and distance will further improve. Note that if we consider $N=1$, resembling an LSA, then this localization method cannot be performed, since no sample covariance matrix exists. At least $N=2$ antennas are required for angle estimation. There is a trade-off between complexity and performance, as adding more antennas or increasing the grid resolutions improves the localization performance at the expense of higher computational complexity. 
In Fig.~\ref{F_locs}(b), the number of antennas in each ULA is fixed at
$N=16$, while the number of arrays $L$ is varied. 
The localization performance improves as $L$ increases, but saturates relatively quickly since the total aperture length is fixed.

\subsection{Channel Estimation}

The localization results can be used to estimate the user's channel. Let us now consider the complete received signal at the  MLA:
\begin{equation}
    \bar{\vect{y}}[\tau] = \sqrt{P \beta}\vect{b}(\varphi_t,d_t) u[\tau] + \vect{n}[\tau],
\end{equation}
where the near-field array response vector is given as
\begin{align} \nonumber
&\vect{b}(\varphi_t,d_t)= \\
&\big[ e^{-\imagunit \frac{2\pi}{\lambda} d_1^{(1)}},\cdots, e^{-\imagunit \frac{2\pi}{\lambda} d_{N}^{(1)}},\cdots, e^{-\imagunit \frac{2\pi}{\lambda} d_1^{(L)}},\cdots, e^{-\imagunit \frac{2\pi}{\lambda} d_N^{{(L)}}}  \big]^{\Ttran}, \label{eq:near-field-bvector}
\end{align} 
and $d_n^{(\ell)}  =  \sqrt{d_t^2 + {\left(\bar{x}_n^{(\ell)}\right)^2} - {2 \bar{x}_n^{(\ell)}{d_t}} \cos(\varphi_t)}  $ denotes the distance between the user and antenna $n$ in ULA $\ell$.

As we have already obtained an estimation of the user location parameters $(\hat{\varphi}_t,\hat{d}_t)$ in the previous subsection, we can estimate the response vector of the array as $\vect{\hat{b}}(\hat{\varphi}_t,\hat{d}_t) = [\hat{b}_1,\dots,\hat{b}_{NL}]^{\Ttran}$ and the large-scale channel gain $\beta = \lambda^2/\left(4 \pi d_t \right)^2$. Therefore, we are now able to recover the channel between the MLA and the user.

\begin{figure}
    \centering
    \includegraphics[width=0.9\linewidth]{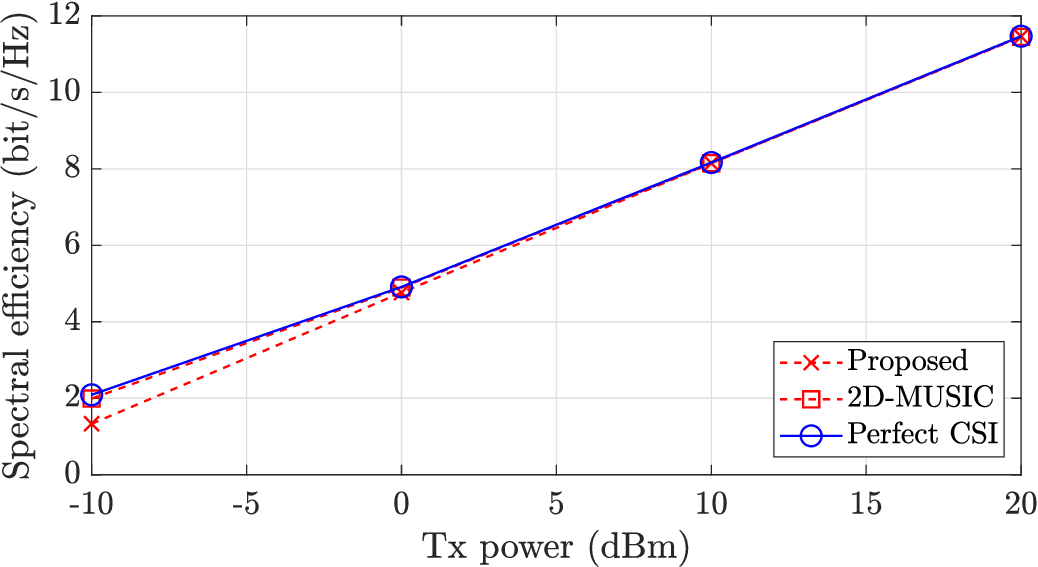}
    \caption{Comparison of the proposed method and $2$D-MUSIC. We set $L=4$, $N=16$, and $D_{\rm array} = 2$\,m.}
    \label{F_ChEstim}
\end{figure}

For a baseline comparison, we consider the $2$D-MUSIC-based approach. {The baseline 2D-MUSIC uses centralized joint processing: all sub-array snapshots are stacked into a single covariance matrix and the MLA steering vector is used, treating the MLA as one array. The user location is then obtained from the peak of the 2D-MUSIC spectrum \cite{2024_Ramezani_Arxiv}.} The two dimensions refer to the angular and distance domains.
The 2D-MUSIC spectrum is 
\begin{equation}
    S_{2M}(\varphi,d) = \frac{1}{\vect{b}^{\Htran}(\varphi,d) \widehat{\vect{U}}_{\rm n}\widehat{\vect{U}}_{\rm n}^{\Htran}\vect{b}(\varphi,d)}. 
\end{equation}
Distance and angle estimates are obtained by identifying the peak within the MUSIC spectrum. These estimates can then be used to determine the channel between the user and the MLA, by substituting the estimated user's location parameters into the channel model, similarly to the proposed approach. 

Channel estimation ultimately affects the achievable SE. For an MF-based combiner, the uplink SE can be computed as $  \text{SE} = \log_2\left(1+\frac{P\beta}{\sigma^2} \left|\frac{\hat{\vect{h}}^{\Htran}}{\|\hat{\vect{h}}\|}  \vect{h}\right|^2\right)$\footnote{The SE expression assumes that the combiner is based on an imperfect channel estimate, but the effective channel $\hat{\vect{h}}^{\Htran} \vect{h} /\|\hat{\vect{h}}\|$ is known during the demodulation (e.g., thanks to a demodulation pilot).}, where $\hat{\vect{h}} = \vect{\hat{b}}(\hat{\varphi},\hat{d})$. In Fig.~\ref{F_ChEstim}, we evaluate the SE for the proposed approach by comparing it with the channel estimation obtained based on the $2$D-MUSIC approach. The system setup is identical to that in Fig.~\ref{F_locs}.
For $2$D-MUSIC, the angular domain is discretized with the same resolution as in the proposed approach, $0.002$ radians, while the distance domain is discretized with a resolution of $0.02$ meters.
When the transmit power is sufficiently high (e.g., above $0$ dBm), the SE achieved by the proposed method closely matches that of the $2$D-MUSIC. We observe a slight gap between the performance of the proposed method and $2$D-MUSIC, in the low SNR regime. This can be viewed as a trade-off between performance and complexity, introduced by sub-array-based processing. Nevertheless, the SE performance of the proposed method is close to that of the ideal scheme (i.e., perfect channel state information),  while the reduction in computational complexity compared to the 2D-MUSIC baseline is substantial, approximately three orders of magnitude. In conclusion, the proposed approach achieves near-optimal performance with significantly reduced complexity.

\begin{remark}
    Subspace-based AoA estimation can be fragile in low-SNR regimes and computationally demanding if sub-arrays have many antennas. Beam training offers a practical alternative for modular arrays where multi-beam training can be implemented by
activating one antenna in each module \cite{zhou2024multi}.
\end{remark}

\section{Conclusions}
\label{Sect_conclude}
This study provided a comprehensive analysis of near-field beamfocusing using an MLA. We highlighted its ability to perform efficient beamfocusing with significantly fewer antennas than with an equally long half-wavelength-spaced ULA and without the ambiguities (e.g., grating lobes) that conventional sparse ULAs suffer from. Specifically, we characterized the beamfocusing behavior both analytically and through simulations, considering factors such as the number of ULAs, the number of antennas per ULA, and the total aperture length of the MLA. The analytical expressions derived for the beamwidth and beamdepth offered valuable insights into the beam characteristics of MLAs composed of two ULAs as well as multiple ULAs. The proposed MLA architecture constitutes a novel BS deployment strategy for 6G networks, where the telecom operator can achieve beamfocusing features using a few normal-sized arrays positioned several meters apart on the same rooftop. This approach can be seen as an initial step toward realizing the futuristic cell-free mMIMO architecture, where the arrays are envisioned to be deployed distributively over the coverage area. Even if the user is in the far-field of the individual ULAs, the MLA combines the ULAs to obtain the depth perception that enables beamfocusing.
This effect resembles how humans perceive depth using two eyes. We demonstrated how an MLA can provide accurate localization both in angle and distance. We proposed a localization method focused on angular estimation in each ULA to significantly reduce computational complexity compared to conventional methods. In general, the results underscored the potential of the MLA architecture to improve the efficiency of future BS deployments.
\appendices

\section{Proof of Lemma~\ref{Fresnel_Approx_ULA}}
\label{App_Fresnel_Approx}
From \eqref{eq:MF_near}, we have
\begin{align}
\hat{G}_{{\rm ULA}}  &=\frac{1}{(NA)^2} \times \nonumber\\
&\left\vert \sum_{n=1}^N \int_{\mathcal{A}_n} e^{+\imagunit\frac{2\pi}{\lambda}\left(\frac{x^2}{2F}+\frac{y^2}{2F}\right)}e^{-\imagunit\frac{2\pi}{\lambda}\left(\frac{x^2}{2z}+\frac{y^2}{2z}\right)}   dx dy\right \vert^2  
\nonumber\\
&=\frac{1}{{{(NA)}^{2}}} \times \notag
\\ \notag
&  \left| \int\limits_{-\frac{N\delta}{2} }^{\frac{N\delta}{2} }\int\limits_{-\frac{\delta}{2}}^{\frac{\delta}{2}} 
e^{\imagunit \frac{2\pi}{\lambda} \left(\frac{x^2}{2F} + \frac{y^2}{2F}\right)}  {e}^{-\imagunit\frac{2\pi }{\lambda }\left( \frac{{{x}^{2}}}{2z}+\frac{{{y}^{2}}}{2z} \right)}  dydx \right|^{2} \\ 
&    =\frac{1}{{{(NA)}^{2}}}{{\left| \int\limits_{-\frac{N\delta}{2} }^{\frac{N\delta}{2} }\int\limits_{-\frac{\delta}{2}}^{\frac{\delta}{2}}   \!\!\! {{{e}^{-\imagunit\frac{2\pi }{\lambda } \frac{\left( |F-z| \right)\left( {{x}^{2}}+{{y}^{2}} \right)}{2zF} }}} dydx \right|}^{2}}.
\end{align}
By defining $z_{\rm eff} = \frac{Fz}{|F-z|}$, we can rewrite the expression as
\begin{align}\label{G-derivation}
    \hat{G}_{\rm ULA} =\frac{1}{(NA)^{2}}{{\left| \int\limits_{-\frac{N\delta}{2} }^{\frac{N\delta}{2} }\int\limits_{-\frac{\delta}{2}}^{\frac{\delta}{2}} 
    {e}^{-\imagunit\frac{\pi }{\lambda } \frac{ {{x}^{2}}}{z_{\rm eff}}  }  {e}^{-\imagunit\frac{\pi }{\lambda } \frac{ {{y}^{2}}}{z_{\rm eff}}  }dydx \right|}^{2}}.
\end{align}
The evaluation of the anti-derivatives in \eqref{G-derivation} yields  \cite{1956_Polk_TAP}
\begin{equation}
\hat{G}_{\rm ULA} =
\frac{\left( {{C}^{2}}\left( \sqrt{ a } \right)+{{S}^{2}}\left( \sqrt{ a } \right) \right) \left( {{C}^{2}}\left( \sqrt{ a } N \right)+{{S}^{2}}\left( \sqrt{ a } N \right) \right)}{( N a )^{2}},
\end{equation}
where $C\left(\cdot \right)$ and $S\left(\cdot \right)$ are the Fresnel integrals, and $ a = \frac{\lambda}{8{z}_{\rm eff}}$. This completes the proof.

\section{Proof of Theorem~\ref{Fresnel_Approx_twoULAs}}
\label{App_Fresnel_Approx_twoULAs}

From \eqref{eq:G2app}, we have
\begin{align}
\notag
& \hat{G}_{2,z_t} = \frac{1}{{{(2  N \delta^2)}^{2}}}  \\ \notag
& \Bigg|   
\int\limits_{-\frac{N \delta}{2} }^{\frac{N\delta}{2} }
\int\limits_{-\frac{ \delta}{2}}^{\frac{ \delta}{2}} 
e^{\imagunit \frac{2\pi}{\lambda} \left(\frac{(x-\overline{\Delta})^2}{2F} + \frac{y^2}{2F}\right)}
{e}^{-\imagunit\frac{2\pi }{\lambda }\left(\frac{{{(x-\overline{\Delta} )}^{2}}}{2z_t}+\frac{{{y}^{2}}}{2z_t} \right)} dxdy  
\\ \notag
& +  \int\limits_{-\frac{N\delta}{2} }^{\frac{N\delta}{2} }\int\limits_{-\frac{ \delta}{2}}^{\frac{ \delta}{2}} 
e^{\imagunit \frac{2\pi}{\lambda} \left(\frac{(x+\overline{\Delta})^2}{2F} + \frac{y^2}{2F}\right)}
{e}^{-\imagunit\frac{2\pi }{\lambda }\left(\frac{{(x + \overline{\Delta} )^{2}}}{2z_t}+\frac{{{y}^{2}}}{2z_t} \right)}  dxdy \Bigg|^2
\\ 
&= \frac{ \Bigg| \int\limits_{-\frac{ \delta}{2}}^{\frac{ \delta}{2}} 
{e}^{\imagunit\frac{\pi }{\lambda }\frac{{{y}^{2}}}{z_{\rm eff}} } dy \left( \int\limits_{-\frac{N\delta}{2}}^{\frac{N\delta}{2} }
e^{\imagunit \frac{\pi}{\lambda} \frac{(x-\overline{\Delta})^2}{ z_{\rm eff}} } dx +
\int\limits_{-\frac{N\delta}{2} }^{\frac{N\delta}{2} }
e^{\imagunit \frac{\pi}{\lambda} \frac{(x+\overline{\Delta})^2}{ z_{\rm eff}} }
 dx \right) \Bigg|^2}{{{(2  N \delta^2)}^{2}}} . \label{G2a}
\end{align} 
The evaluation of the anti-derivatives in \eqref{G2a} using the Fresnel integrals yields
\begin{align}
\notag
\hat{G}_{2,z_t} = &\frac{1}{{{(4  N a)}^{2}}}  
\left( C^2(\sqrt{a}) + S^2(\sqrt{a}) \right)  
\\ \notag
& \cdot\Big( \left( C(\beta_1) - C(-\beta_1) +C(\beta_2) - C(-\beta_2) \right)^2 + \\ 
& \left( S(\beta_1) -S(-\beta_1) + S(\beta_2) -S(-\beta_2) \right)^2\Big),
 \label{G2b}
\end{align} 
where  $a$, $\beta_1$, and $\beta_2$ are defined in the theorem. Since the Fresnel integral functions $C(\cdot)$ and $S(\cdot)$ are both odd functions, we can simplify the expression as
\begin{multline}
 \hat{G}_{2,z_t} = \frac{1}{{{(2  N a)}^{2}}}  
 \left( C^2(\sqrt{a}) + S^2(\sqrt{a}) \right)\cdot \\ 
\Big( \left( C(\beta_1) +C(\beta_2) \right)^2 + \left( S(\beta_1)  + S(\beta_2)  \right)^2\Big),
\end{multline} 
which completes the proof.

\bibliographystyle{IEEEtran}
\newpage

\bibliography{IEEEabrv,refs}
\begin{IEEEbiography}
[{\includegraphics[width=1in,height=1.25in,clip,keepaspectratio]
{{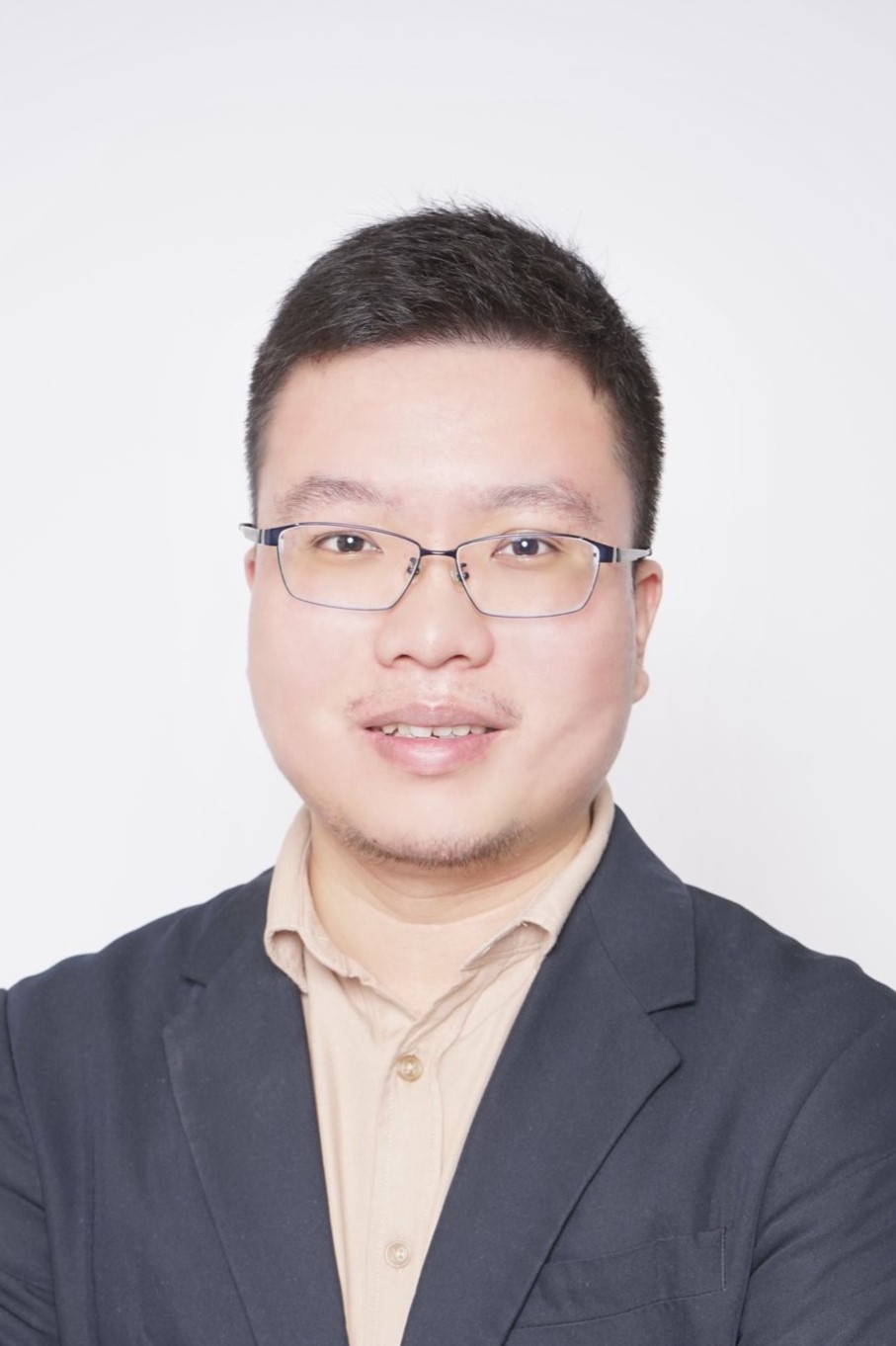}}}]{Alva Kosasih}
received the B.Eng. and M.Eng. degrees in electrical engineering from Brawijaya University, Indonesia, in 2013 and 2017, respectively, the M.S. degree in communication engineering from National Sun Yat-sen University, Taiwan, in 2017, and the Ph.D. degree in communication engineering from the University of Sydney, Australia, in 2023. From 2023 to 2024, he was a Postdoctoral Researcher with the KTH Royal Institute of Technology, Stockholm, Sweden. He is currently a 3GPP RAN1 Delegate with Nokia Standards. His research interests include signal processing for large-scale MIMO; near-field beamfocusing, localization, and channel estimation; MIMO symbol detection; and machine learning for physical-layer communications
\end{IEEEbiography}

\begin{IEEEbiography}
[{\includegraphics[width=1in,height=1.25in,clip,keepaspectratio]
{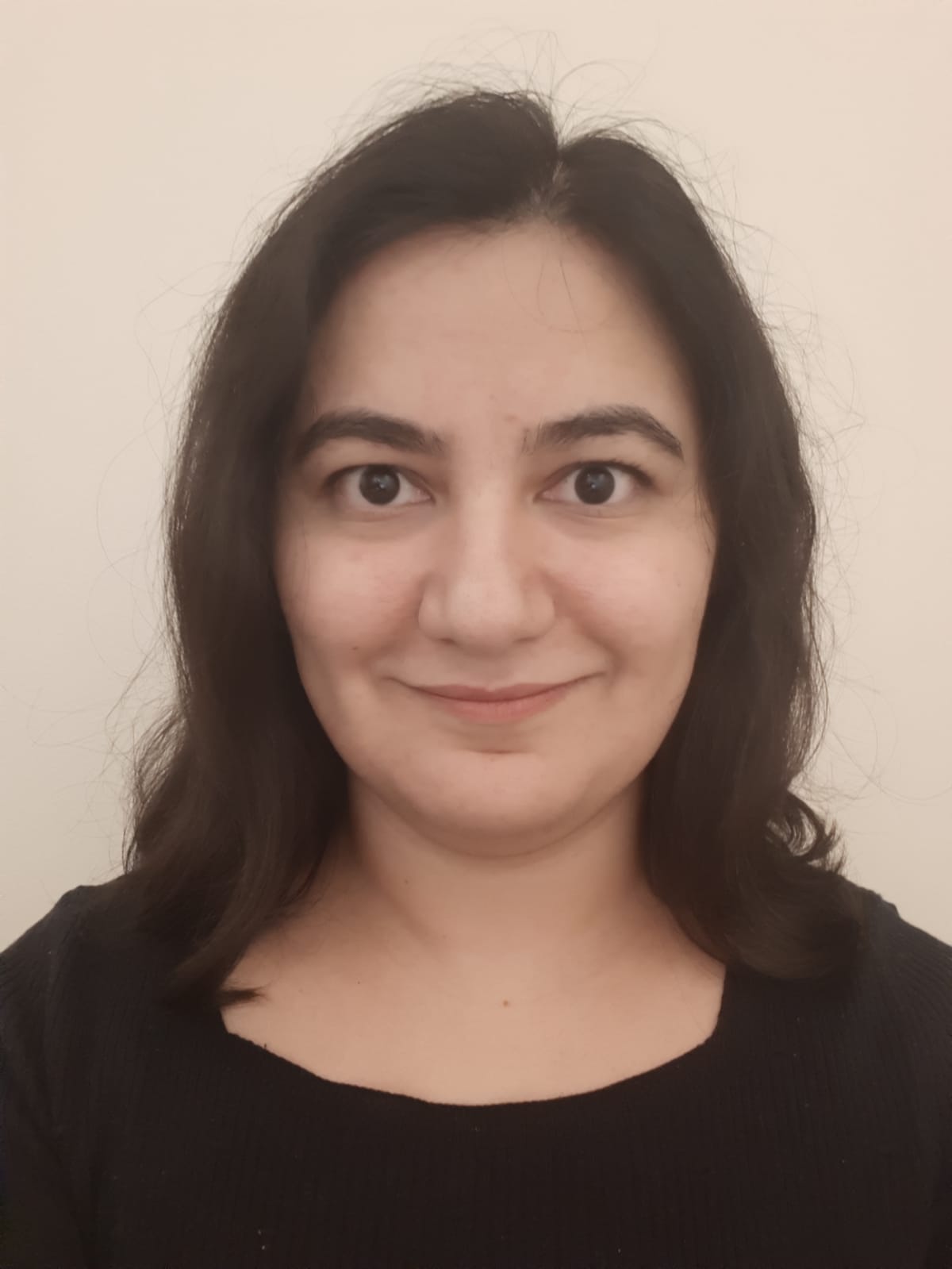}}]{\"Ozlem Tu\u{g}fe Demir }
is an Assistant Professor of Electrical and Electronics Engineering at Bilkent University, Ankara, Turkiye. She was with TOBB University of Economics and Technology, Ankara, Turkiye. She received the B.S., M.S., and Ph.D. degrees in Electrical and Electronics Engineering from Middle East Technical University, Ankara, Turkey, in 2012, 2014, and 2018, respectively. She was a Postdoctoral Researcher at Linköping University, Sweden in 2019-2020 and at KTH Royal Institute of Technology, Sweden in 2021-2022. She has co-authored the textbooks Foundations of User-Centric Cell-Free Massive MIMO (2021) and Introduction to Multiple Antenna Communications and Reconfigurable Surfaces (2024). Her research interests focus on signal processing and optimization in wireless communications, massive MIMO, cell-free massive MIMO, beyond 5G multiple antenna technologies, reconfigurable intelligent surfaces, near-field communications, and green mobile networks.
\end{IEEEbiography}

\begin{IEEEbiography}
[{\includegraphics[width=1in,height=1.25in,clip,keepaspectratio] {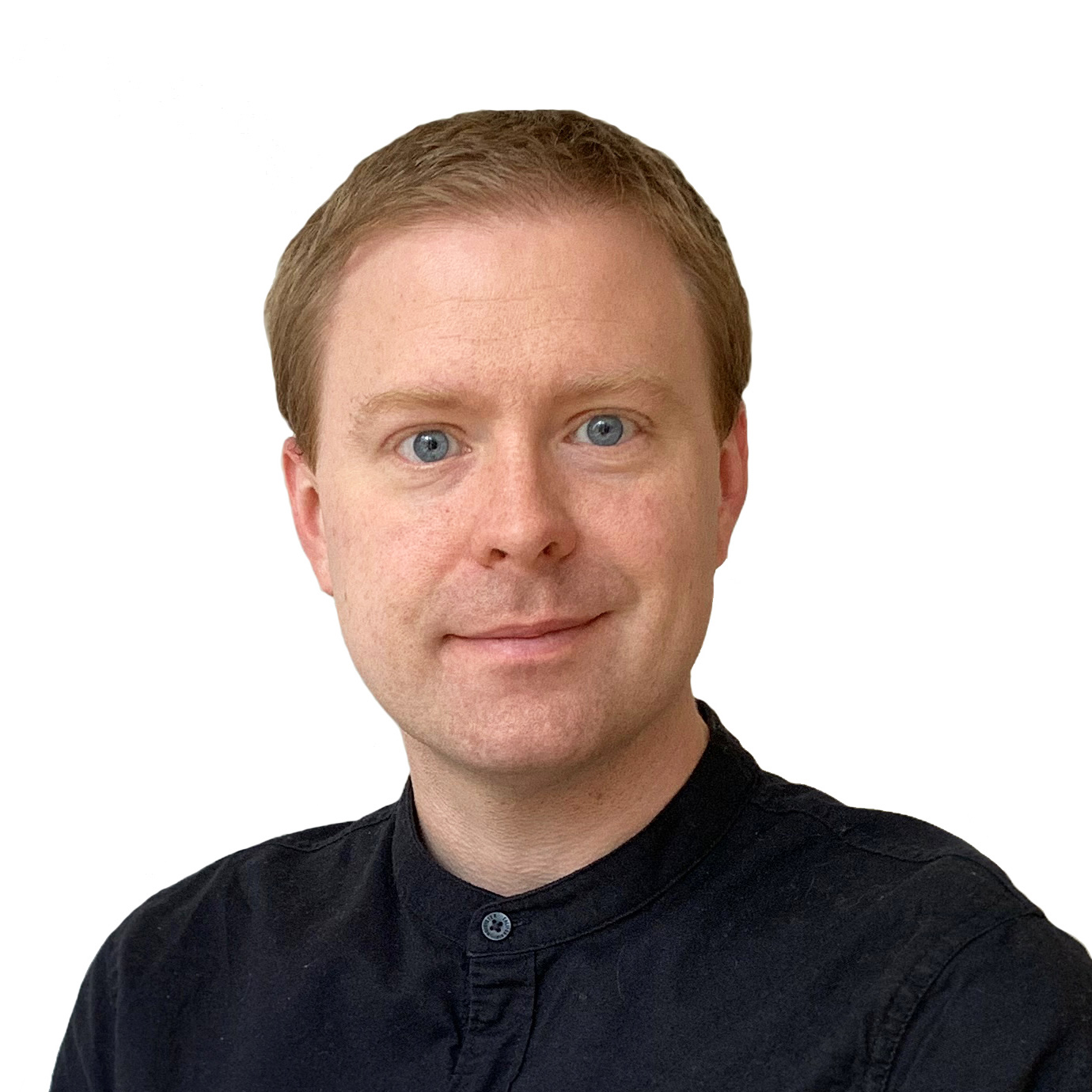}}]{Emil Bj\"ornson }
received the M.S. degree in engineering mathematics from Lund University, Sweden, in 2007, and the Ph.D. degree in telecommunications from the KTH Royal Institute of Technology, Sweden, in 2011.

From 2012 to 2014, he was a Post-Doctoral Researcher with the Alcatel-Lucent Chair on Flexible Radio, SUPELEC, France. From 2014 to 2021, he held different professor positions at Link\"oping University, Sweden. He has been a Full Professor of Wireless Communication at KTH since 2020 and the Head of the Communication Systems division since 2024. He has authored the textbooks \emph{Optimal Resource Allocation in Coordinated Multi-Cell Systems} (2013), \emph{Massive MIMO Networks: Spectral, Energy, and Hardware Efficiency} (2017), \emph{Foundations of User-Centric Cell-Free Massive MIMO} (2021), and \emph{Introduction to Multiple Antenna Communications and Reconfigurable Surfaces} (2024). He is dedicated to reproducible research and has published much simulation code. He researches multi-antenna communications, reconfigurable intelligent surfaces, radio resource allocation, machine learning for communications, and energy efficiency.

Dr. Bj\"ornson has performed MIMO research since 2006. His papers have received more than 40000 citations, he has filed more than 30 patent applications, and he is recognized as a Clarivate Highly Cited Researcher. He co-hosts the podcast Wireless Future and has a popular YouTube channel with the same name. He is a Wallenberg Academy Fellow, a Digital Futures Fellow, and an SSF Future Research Leader. He has received the 2014 Outstanding Young Researcher Award from IEEE ComSoc EMEA, the 2015 Ingvar Carlsson Award, the 2016 Best Ph.D. Award from EURASIP, the 2018 and 2022 IEEE Marconi Prize Paper Awards in Wireless Communications, the 2019 EURASIP Early Career Award, the 2019 IEEE ComSoc Fred W. Ellersick Prize, the 2019 IEEE Signal Processing Magazine Best Column Award, the 2020 Pierre-Simon Laplace Early Career Technical Achievement Award, the 2020 CTTC Early Achievement Award, the 2021 IEEE ComSoc RCC Early Achievement Award, the 2023 IEEE ComSoc Outstanding Paper Award, and the 2024 IEEE ComSoc Stephen O. Rice Prize. He also coauthored papers that received best paper awards at five conferences.

\end{IEEEbiography}
\end{document}